\documentclass[a4paper,11pt]{article}

\usepackage{jheppub} 
\usepackage{dsfont}
\usepackage{extarrows}
\usepackage{hyperref}
\usepackage{cleveref}
\usepackage{supertabular}
\usepackage{todonotes}
\usepackage{mathrsfs}
\usepackage{array}
\usepackage{lipsum}
\usepackage{physics}

\usepackage{slashed}
\usetikzlibrary{decorations.markings}

\usepackage{helvet} 
\usepackage{amsmath}
\usepackage{amssymb} 
\usepackage{amsthm}
\usepackage{url}
\usepackage{tikz}
\usepackage{tikz-cd}
\usetikzlibrary{arrows}
\usetikzlibrary{arrows.meta}
\usetikzlibrary{shapes.geometric,calc,arrows, positioning,shapes.misc,decorations.markings}

\newcommand{\newreptheorem}[2]{%
\newenvironment{rep#1}[1]{%
 \def\rep@title{#2 \ref{##1}}%
 \begin{rep@theorem}}%
 {\end{rep@theorem}}}
\makeatother
\newtheorem{lemma}{Lemma}[section]
\newreptheorem{lemma}{Lemma}
\newtheorem{theorem}[lemma]{Theorem}
\newtheorem{corollary}[lemma]{Corollary}
\newtheorem{prop}[lemma]{Proposition}

\newreptheorem{conj}{Conjecture}
\newtheorem{question}[lemma]{Question}

\newtheorem{defn}[lemma]{Definition}

\newtheorem{remark}[lemma]{Remark}

\newtheorem{exm}[lemma]{Example}

\usepackage{lscape} 
\usepackage{braket}

\usepackage[T1]{fontenc} 

\newcommand{\CH}{\Ch_{\mu_1\mu_2}(Q_F,{t},{q})}

\newcommand{\Rep}{\textbf{Rep}}

\newcommand{\be}{\begin{equation}}
\newcommand{\ee}{\end{equation}}
\newcommand{\ba}{\begin{aligned}}
\newcommand{\ea}{\end{aligned}}
\newcommand{\bea}{\begin{eqnarray}}
\newcommand{\eea}{\end{eqnarray}}

\def\mb{\mathbb}

\def\mc{\mathcal}

\def\bp{\begin{pmatrix}}
\def\ep{\end{pmatrix}}

\def\bbC{\mathbb{C}}

\def\bbG{\mathbb{G}}

\def\bbK{\mathbb{K}}

\def\bbN{\mathbb{N}}

\def\bbQ{\mathbb{Q}}
\def\bbR{\mathbb{R}}

\def\bbZ{\mathbb{Z}}
\def\CA{\mathcal{A}}

\def\CC{\mathcal{C}}

\def\CG{\mathcal{G}}
\def\CH{\mathcal{H}}

\def\CK{\mathcal{K}}
\def\CL{\mathcal{L}}
\def\CM{\mathcal{M}}

\def\CT{\mathcal{T}}

\def\CZ{\mathcal{Z}}
\def\mfr{\mathfrak}

\def\calg{$C^*$-algebra\ }
\def\calgs{$C^*$-algebras\ }

\def\hilb{{\rm Hilb}}

\def\vecg{{\rm Vec}_G}
\def\vecgo{{\rm Vec}_G^\omega}
\def\gmg{G//_{\rm Ad} G}

\title{\boldmath Categorical Symmetries via Operator Algebras}

\author[b]{Qiang Jia}
\author[c]{Ran Luo}
\author[a]{Jiahua Tian}
\author[c,d]{Yi-Nan Wang}
\author[e]{Yi Zhang}
\affiliation[a]{School of Physics, East China Normal University, \\
Shanghai, China, 200241}
\affiliation[b]{Department of Physics, Korea Advanced Institute of Science Technology, \\
Daejeon 34141, Korea}
\affiliation[c]{School of Physics, Peking University,\\
Beijing, China, 100871}
\affiliation[d]{Center for High Energy Physics, Peking University, \\
Beijing, China, 100871}
\affiliation[e]{Kavli IPMU (WPI), UTIAS, The University of Tokyo, \\
Kashiwa, Chiba 277-8583, Japan}
\emailAdd{qjia1993@kaist.ac.kr}\emailAdd{jtian1905@gmail.com}\emailAdd{ynwang@pku.edu.cn}
\emailAdd{ranluo@pku.edu.cn}\emailAdd{yi.zhang@ipmu.jp}

\abstract{We propose that the symmetry category associated to a 2D quantum field theory with 0-form $G$-symmetry with 't Hooft anomaly $k\in H^4(BG,\mathbb{Z})$ for a large class of Lie groups $G$ is the category of twisted measurable fields of Hilbert spaces over $G$ denoted by $\mathrm{Hilb}^k(G)$, which is equivalent to the category of unitary representations of $C_0(G)$ with convolution product twisted by a multiplicative bundle gerbe labeled by $k$ denoted by $\Rep^k(C_0(G))$. We find that the Drinfeld center of the symmetry category $\mathcal{Z}(\mathrm{Hilb}^{k}(G))$ equivalent to the category of unitary representations of the groupoid $C^*$-algebra of the Fell line bundle $\Sigma_k$ over the conjugation action groupoid $G//_{\rm Ad} G$, denoted by $\Rep(C^*(G//_{\rm Ad}G,\Sigma_k))$, where the twist is characterized by the transgression $\tau(k)\in H^2(G//_{\rm Ad}G,U(1))$. To the full generality, our framework applies to a Lie group $G$ that is a direct product of a compact connected Lie group and a number of $\mathbb{R}$ or $GL(1,\mathbb{C})$ factors. We compute the braiding of anyon lines in the bulk 3D SymTFT from this formalism. Finally we provide physical examples for abelian and non-abelian $G$, and discuss the physical consequences of flat gauging continuous global symmetries.}

\begin{document}
\maketitle
\flushbottom

\section{Introduction}

Despite the immense success of the SymTFT or the Topological Holography paradigm in describing both the finite group-theoretic and finite semisimple non-invertible symmetries~\cite{Witten:1998wy,Kong:2015flk,He:2016xpi,Ji:2019jhk,Kong:2020cie,Gaiotto:2020iye,Apruzzi:2021nmk,DelZotto:2022ras,Moradi:2022lqp,Ji:2021esj,Freed:2022qnc,Kaidi:2022cpf,vanBeest:2022fss,Kaidi:2023maf,Hai:2023osv,Bhardwaj:Lecture,Schafer-Nameki:2023jdn,Bhardwaj:2023fca,Huang:2023pyk,Cao:2023rrb,Bhardwaj:2023idu,Bhardwaj:2023bbf,Choi:2024tri,Luo:2025phx,Schafer-Nameki:2025fiy,Qi:2025tal,Bergman:2026lnz}, a complete, satisfactory categorical approach for continuous symmetries is still unestablished. 

On the physics side, it is natural to expect that the SymTFT for continuous symmetries is described by certain non-compact BF theory with a non-compact $B$ field, see~\cite{Brennan:2024fgj,Antinucci:2024zjp,Bonetti:2024cjk,Apruzzi:2024htg,Antinucci:2024bcm,Gagliano:2024off,Argurio:2024ewp,Jia:2025jmn,Bonetti:2025dvm} for the discussions of SymTFT action for abelian, non-abelian Lie group symmetries, for the application to spacetime (super-)symmetries~\cite{Apruzzi:2025hvs,Borsten:2025pvq,Ambrosino:2026hhv}, and the non-abelian group with anomaly~\cite{ST_Kernel}.

However, on the categorical side, a systematic formulation for the Lie group symmetry categories and an explicit calculation of their Drinfeld centers remains obscure. It is important to check the validity of the slogan \textit{SymTFT equals to the Drinfeld center of the symmetry category} $\mc{S}$ in the case of continuous Lie group symmetries, before venturing into the even wilder zoo of continuous non-invertible symmetries\cite{Bhardwaj:2022yxj,GarciaEtxebarria:2022jky,Antinucci:2022eat,Damia:2023gtc,Hsin:2024aqb,Hsin:2025ria,Delmastro:2025ksn}.

The major mathematical obstruction is to formulate the infinitely many simple objects with the finesse of incorporating the standard topology associated to the Lie group manifold, going beyond the realm of fusion categories. There has been two proposals for Lie group symmetry categories: category of skyscraper sheaves~\cite{Freed:2009qp,Jia:2025vrj} ${\rm Sky}(G)$ and category of quasi-coherent sheaves~\cite{Jia:2025vrj,Stockall:2025ngz} ${\rm QCoh}(G)$, each with its own failure. The skyscraper proposal ${\rm Sky}(G)$ does not accommodate infinite direct sums of simple objects, rendering important $d\ge 1$ submanifolds of $G$ (e.g. conjugacy classes) out of its capabilities. The quasi-coherent sheaf proposal ${\rm QCoh}(G)$ allows for infinite direct sums of simple objects, but most mathematical theorems we desire are applicable only if $G$ is an algebraic group, which is not generic enough.

In this work, we construct the symmetry category $\hilb^k(G)$ and its Drinfeld center $\CZ(\hilb^k(G))$ for compact Lie group $G$ in $2$D with anomaly $k\in H^4(BG,\bbZ)$. The centerpiece of our construction involves representations of $C^*$-algebra. Inspired by previous mathematical formulations of continuous tensor category~\cite{Marin-Salvador:2025stc} and manifold tensor category~\cite{Weis:2022egw}, we construct the a continuous tensor category with infinite number of simple objects. $\hilb^k(G)$ as representations of the twisted function space $C_0(G,\rho)$. The calculation of Drinfeld center $\CZ(\hilb^k(G))$ is given by a Fell bundle $\Sigma_k$ over the inertia groupoid $\gmg$ encoding the anomaly information. Finally, we are able to classify the simple objects in $\CZ(\hilb^k(G))$ as irreducible representations of the $C^*$-algebra associated to $\Sigma_k$.

The main mathematical result of our work is Proposition~\ref{prop:centerclassify}, the simple objects in Drinfeld center $\CZ(\hilb^k(G))$ are labeled by 
\be ([g], R )  \ee
where $[g]$ is a conjugacy class of $G$ and $R$ is an irreducible projective representation of the centralizer $C_G(g)$.

The physical conclusion of our work is clear: the symmetry category for an anomalous Lie group $G$ symmetry is $\hilb^k(G)$ with $k\in H^4(BG,\mb{Z})$ and its SymTFT is the Drinfeld center $\mc{Z}(\hilb^k(G))$ with simple objects described above, given that $G$ is a direct product of a finite number of~\footnote{There are two conditions on $G$, the first one is the injectivity of the transgression $\tau:H^4(BG,\mb{Z})\rightarrow H^3_G(G,\mb{Z})$, and the second one is the amenability condition on $G$ (automatically satisfied for all compact groups).}
\begin{itemize}
    \item compact connected Lie groups, or
    \item non-compact, connected, amenable, reductive algebraic groups (over $\mathbb{R}$ or $\mathbb{C}$), which only include the $\mb{R}\cong GL(1,\mb{R})_+$ or $GL(1,\mb{C})$ factors.
\end{itemize}
In this framework, the continuous infinite direct sum of simple objects are included in the categories $\hilb^k(G)$ and $\mc{Z}(\hilb^k(G))$, hence it is mathematically sensible to define different boundary conditions of the SymTFT, such as the ones with $G$ or Rep$(G)$ symmetries.

\paragraph{Structure of this paper:} 

In Section~\ref{sec:fusioncat}, we review the fusion categories $\vecg$ and $\vecgo$ of finite group $G$ and their Drinfeld centers in the language of $C^*$-algebras, providing intuition for the continuous case. The non-anomalous category $\vecg$ is realized as the representation category of the Hopf $C^*$-algebra $C(G)$, and its Drinfeld center is derived via the inertia groupoid. For the anomalous category $\vecgo = \Rep (C^\omega(G))$, we incorporate a twist via a Fell bundle to construct its Drinfeld center.

In Section~\ref{sec:non-anomalous}, we construct the Lie group symmetry category $\hilb(G)$ and its Drinfeld center. We introduce the non-anomalous Lie group symmetry category $\hilb(G)$ as the representation category of the Hopf $C^*$-algebra $C_0(G)$. Then we present a calculation of $\CZ(\hilb(G))$ through the construction of quantum double $DC_0(G)$, and demonstrate the equivalence with Fell bundle construction.

In Section~\ref{sec:anomalous}, we construct the anomalous symmetry category $\hilb^k(G)$ and its Drinfeld center. We first demonstrate how the anomaly $k\in H^4(BG,\bbZ)$ could be transgressed into the Fell bundle data. Then we construct the symmetry category as a continuous tensor category from \calg $C_0(G,\rho)$ and correspondences. Thereafter we explain our intuition of constructing $\CZ(\hilb^k(G))$ as the representation category associated to the Fell bundle \calg $C^*(\gmg,\Sigma_k)$, and calculate the irreducible representations of $C^*(\gmg,\Sigma_k)$, which are simple objects of $\CZ(\hilb^k(G))$.

In Section~\ref{sec:physicalexm}, we apply continuous SymTFT on $2$D scalar field theories with an explicit focus on flat gauging transitions between different symmetry sectors. In a $2$D compact scalar theory, we demonstrate the SymTFT of anomalous $U(1)_m \times U(1)_w$ global symmetry from the Fell bundle construction. Through different condensations, the theory can be transformed either into a non-compact scalar with $\mathbb{R}$ symmetry or into a compact scalar with a modified radius. At the self-dual radius $R=1$ with the enhanced $SO(4)$ symmetry, flat gauging a diagonal $SO(3)$ subgroup results in the Runkel-Watts model, a non-rational $c=1$ CFT.

In Section~\ref{sec:outlook}, we discuss the possible generalizations and applications into more generic physics scenarios. Also, we put forward a series of open mathematical questions for a rigorous definition of such SymTFT and the classification of boundary conditions.

The appendices provide the definitions for most mathematical objects we applied in the main text. In Appendix~\ref{app:OperA}, we introduce the basic notions of $C^*$-algebra and representation. In Appendix~\ref{app:bimod}, we introduce Hilbert (bi)modules of $C^*$-algebras and correspondences between $C^*$-algebras. In Appendix~\ref{app:continuoustc}, we introduce the continuous tensor category proposed in~\cite{Marin-Salvador:2025stc}, and state how the data of $C^*$-algebra and correspondences manifests as the categorical structure. In Appendix~\ref{app:multbdgb}, we introduce multiplicative bundle gerbes and its classification. In Appendix~\ref{app:fellbundle}, we introduce the definitions and constructions for Fell bundle, the associated $C^*$-algebra and its representations. In Appendix~\ref{app:injectivetransgression}, we prove that for compact connected Lie groups, the transgression from $H^4(BG,\bbZ)$ to the equivariant cohomology $H^3_G(G,\bbZ)$ is always injective.

\section{Intuition from Fusion Category}\label{sec:fusioncat}
There are various ways to construct a fusion category and the Drinfeld center of it. We will give a brief introduction of these formalisms for finite group symmetries in this section, and then generalize these intuitions into generic Lie groups in subsequent sections. 

Within the scope of this work, we focus on the fusion category ${\rm Vec}^\omega_G$ of $G$-graded vector spaces for finite group $G$ with a twist $\omega\in H^3_{\rm grp}(G,U(1))$ signifying the 't Hooft anomaly of the $G$ 0-form symmetry in the 2D boundary physical system. In this section, we will always assume $G$ to be a finite group.

\subsection{Hopf $C^*$-Algebra Description}
The category $\vecg$ can be equivalently formulated as a representation category of the \calg $C(G)$ of bounded $\bbC$-valued functions over $G$. 

This Abelian \calg $C(G)$ is given by the following data:
\begin{itemize}
    \item a set of elements $\{ \text{bounded functions \ } f:G\to \bbC\}$;
    \item the associative multiplication $(f\cdot f') (g) = f(g)f'(g)$ for every $f,f'\in C(G)$;
    \item the involution $f^*(g) = \overline{f(g)}$
    \item the supremum norm $||f|| = \sup_{g\in G} |f(g)|$, one can easily verify $||f^*\cdot f||= ||f||^2$.
\end{itemize}
According to the results of~\cite{Takesaki:1979,Farah:2019}, a (non-degenerate $*$) representation $\pi \in  \Rep(C(G))$ is given by a map from $C(G)$ to $\mathscr{B}(\CH)$, where $\CH$ is a $G$-graded vector space\footnote{Rigorously speaking, this should be $G$-graded Hilbert space. But for finite group $G$ they are in one to one correspondence.}, for each $\ket{\xi}\in \CH_g$ we have
\be
\pi(f)\ket{\xi}  = f(g)\cdot \ket{\xi} \ .
\ee

The morphisms of $\Rep(C(G))$ are intertwiners between representations, which are equivalent to the data of $G$-graded linear maps. By now the category $\Rep (C(G))$ does not admit any monoidal structure. To define a monoidal structure we need to know the tensor product of representations $\pi_i:C(G) \to \mathscr{B}(\CH_i)$. A naive tensor product gives
\be
(\pi_1\otimes \pi_2):  C(G)\otimes C(G) \to \mathscr{B}(\CH_1\otimes\CH_2) \ , \ee
which is not a representation in $\Rep (C(G))$.
A natural way to define such tensor product is through defining a coproduct structure on the algebra $C(G)$
\be\ba
\Delta : C(G) &\to C(G)\otimes C(G) \cong C(G\times G) \\
f   &\mapsto \Delta(f)= \sum f^{(1)}\otimes f^{(2)}:(g,h)\mapsto f(gh) \ ,
\ea\ee
upon which we can give the monoidal structure of representation category
\be
(\pi_1\otimes_{\Rep(C(G))} \pi_2)(f) := (\pi_1\otimes \pi_2)\circ\Delta :C(G) \to C(G)\otimes C(G) \to \mathscr{B}(\CH_1\otimes \CH_2)  \ .
\ee
To make everything consistent, the counit and antipode shall also be assigned
\be\ba
&\varepsilon:C(G) \to \bbC \ , \ \varepsilon(f) = f(e)\\
&S:C(G) \to C(G) \ , \ S(f)(g) = f(g^{-1}) \ .
\ea\ee
These additional structures would render $C(G)$ a Hopf \calg. With the above data, there is a monoidal equivalence of fusion categories $\vecg \cong \Rep(C(G))$~\cite{Etingof:2005,Mason:2017klg,etingof2017tensor,Khovanov:2015acq}.

The Drinfeld center of $\vecg \cong \Rep(C(G))$ can be formulated as the representation category of the extension $C(G)\rtimes_{\rm Ad} G$ by conjugate action 
\be\ba
(g\rhd f) (h) &= f(ghg^{-1})
\ea\ee
This \calg is isomorphic to the Drinfeld double or quantum double, usually denoted as $D(G)$. As a Hopf algebra, $D(G)$ gives product and coproduct best described with the basis $C(G) = {\rm span}\{e_g: g\mapsto 1 ,h(\ne g)\mapsto 0\}$ and the group algebra $\bbC[G] = {\rm span}\{g \}$,
\be\ba
(e_g,x) \cdot (e_h,y) = (\delta_{xgx^{-1},h} e_g ,xy)\\
\Delta(e_g,x) = \sum_{hk = g} (e_h,x) \otimes (e_k,x) \ .
\ea\ee
It is well-known that $\CZ(\vecg) = \Rep(D(G))$~\cite{Etingof:2024abc}. 
Now, it is intuitive to fathom if this framework for symmetry category and Drinfeld center could be generalized to Lie group symmetries, with caution for the infinite dimensionality of Hopf $C^*$-algebra over a smooth base manifold.

\subsection{Inertia Groupoid}
There is another formulation for calculating the Drinfeld center of $\vecg$, known as the tube category or the inertia groupoid~\cite{KAWASAKI197875,Lupercio:2004bro}.

For a group $G$, the inertia groupoid is formally defined by the adjoint action of the group $G$ on itself via conjugation. Concretely, this groupoid (denoted $\CG = G//_{\rm Ad}G$) is a category given by the following data:
    \begin{itemize}
        \item The objects $\CG^{(0)} := {\rm Obj}(\CG)= G$.
        \item The morphisms ${\rm Hom}_{\CG}(g,h) = \{ x\in G| xgx^{-1} = h \}$, if $g$ and $h$ are not in the same conjugacy class, ${\rm Hom}_{\CG}(g,h) = \emptyset$. For each 1-morphism $x\in {\rm Hom}_\CG (g,xgx^{-1})$, we label it with $(g,x)$. We also denote the entire morphism space $\CG^{(1)}:= \sqcup_{g,h} {\rm Hom}(g,h)$.
        
        \item Multiplication as composition of 1-morphisms: $(xgx^{-1},y)\cdot (g,x) = (g,xy) \ , \ \forall g,x,y\in G$.
    \end{itemize}
For any morphism we use $s,t$ to denote the source and the target object. In the case, for any $x = (g,k)\in \CG$ we have $s(g,k) = g$ and $t(g,k) = kgk^{-1}$.

Geometrically, the inertia groupoid provides a rigorous way to understand the free loop space of the classifying space (or classifying stack) $BG$. The isomorphism classes of principal $G$-bundles on a circle $S^1$ are in direct bijection with the connected components of the free loop space $LBG$. 
It is natural to consider its representation category $\Rep(\gmg)$, where each representation is given by a functor $F:\gmg \to {\rm Vec}$. Explicitly, this would encode the data
\be
\begin{tikzcd}
F:\gmg \arrow[rr]             &                    & {\rm Vec}                    \\
g \arrow[rr, "F_g", maps to] \arrow[d, "{(g,h)}"] & & V_g \arrow[d, "{F_{(g,h)}}"] \\
hgh^{-1} \arrow[rr, "F_{hgh^{-1}}"]       &        & V_{hgh^{-1}}                
\end{tikzcd}
\ee
with the additional requirement that every morphism $F_{(g,h)}$ is an invertible linear map, $V_g \cong V_{hgh^{-1}}$. 

Thus the resulting vector space can be regrouped according to the conjugacy classes
\be  V =  \bigoplus_{g\in G} V_g  \cong \bigoplus_{[g]\in {\rm Cl}(G)} (V_g^{\oplus |[g]|} ) \ .  \ee
For an irreducible representation, we ought to set $V_g \cong \bbC$ for one conjugacy class $[g]$ and $V_h = \emptyset ,\ \forall h\notin [g]$. The invertible linear maps form the (induced representations of) irreducible representations of the centralizer $C_G(g)$. Such doublets labeled by $([g],\pi\in\mathrm{IrRep}(C_G(g)))\in \Rep(\gmg)$ exactly one-to-one correspond to the simple objects of the Drinfeld center $\CZ(\vecg)$.

With a closer look on morphisms, it has been verified that the two categories are also equivalent~\cite{Ostrik:2001xnt,Willerton:2008gyk,Ben-Zvi:2008vtm,Fantechi:2001kcr},
\be \Rep(\gmg) \cong \CZ(\vecg) \ .  \ee

As a final remark, for the symmetry of a $2$D physical system described by a fusion category $\mathcal{C}$, its Drinfeld center $\mathcal{Z}(\mathcal{C})$ is known to be equivalent to the representation category $\Rep(\mathrm{Tube}(\mathcal{C}))$ of the tube algebra $\mathrm{Tube}(\mathcal{C})$. The tube algebra $\mathrm{Tube}(\mathcal{C})$, introduced in \cite{ocneanu1994chirality}, is a finite-dimensional $C^*$-algebra, and its relation to the Drinfeld center $\mathcal{Z}(\mathcal{C})$ was established in \cite{evans1995ocneanu,Izumi:2000qa,Muger}; see also \cite{Lin:2022dhv} for a review. Physically, the tube algebra can be understood as the fusion algebra acting on both local operators and twist operators, where the latter are point-like operators attached to topological lines $L$ in $\mathcal{C}$.\footnote{Twist operators are also referred to as defect operators or disorder operators. Under the operator-state correspondence, they are identified with states in the $L$-twisted Hilbert space.} In the special case $\mathcal{C}=\mathrm{Vec}_G$, a twist operator is attached to a tail labeled by a group element $g$, and the action of $G$ on this tail is given by conjugation. Consequently, $\mathrm{Tube}(\mathrm{Vec}_G)$ is naturally identified with $G//_{\text{Ad}}G$ as a groupoid.

\subsection{Twisted Category}
For an anomalous group symmetry $\vecgo$ with nontrivial $\omega$, there are parallel constructions in both languages we introduced above. Analogous to the untwisted case, $\vecgo$ can also be formulated as the representation category of a certain quasi-Hopf $C^*$-algebra. Its Drinfeld center $\CZ(\vecgo)$ can be constructed from multiple equivalent ways~\cite{Schauenburg:1999cnm,Bhowmick:2018}:
\begin{enumerate}
    \item the representation category of the twisted Drinfeld double $D^\omega (G)$~\cite{Roche:1990hs,Drinfeld:1989st};
    \item the representation category of the twisted inertia groupoid~\cite{Moerdijk:2002hiu};
    \item the representation category of $C^*$-algebra associated to the Fell bundle~\cite{Roy:2015}.
\end{enumerate}
We will briefly introduce the first two well-known constructions, and elaborate on the last one, which can be directly generalized into the Lie group case.

The symmetry category $\vecgo$ can be formulated as the representation category of a twisted version of $C(G)$, which is a quasi-Hopf algebra denoted $C^\omega(G)$. As a vector space, this is exactly the same as $C(G)$, but other structures shall conform to the twist $\omega$\footnote{We choose a representative $\omega:G\times G\times G\rightarrow U(1)$ in the cohomology class $H^3_{\rm grp}(G,U(1))$.}. The product and coproduct are the same as $C(G)$. The twist demonstrates its power through the isomorphism $\Phi$
\be
\begin{tikzcd}
&  & (C^\omega(G)\otimes C^\omega(G))\otimes C^\omega(G) \arrow[dd, "\Phi"] \\
C^\omega(G)\otimes C^\omega(G) \arrow[rrd, "{\rm id}\otimes \Delta"'] \arrow[rru, "\Delta \otimes {\rm id}"] &  &    \\
&  & C^\omega(G)\otimes (C^\omega(G)\otimes C^\omega(G))                   
\end{tikzcd}
\ee
\be
\Phi(f,f',f'') = \sum_{g,h,k\in G} \omega(g,h,k)^{-1} f(g) \otimes f'(h) \otimes f''(k) \ .
\ee
Consistency condition would give twists to other structures. For more detailed construction, please refer to~\cite{Roche:1990hs}.

The Drinfeld center can be constructed in similar ways. The \calg is constructed by the conjugate action with a twist, denoted as $C(G)\rtimes^\omega_{\rm Ad} G$. The resulting Drinfeld double $D^\omega(G)$ admits a twist in its product and coproduct structure as well~\cite{Roche:1990hs}
\be\ba
 \left(e_g, x\right) \cdot\left(e_h, y\right) &=\left(\theta_g(x,y)\delta_{x g x^{-1}, h} e_g, x y\right) \\
 \Delta\left(e_g, x\right) &=\sum_{h k=g} \gamma_x(h,k)\left(e_h, x\right) \otimes\left(e_k, x\right)
\ea\ee
with the twist phases given by
\be\begin{aligned}
\theta_g(x, y) & :=\frac{\omega(g, x, y) \omega\left(x, y, (xy)g(xy)^{-1} \right)}{\omega\left(x, xgx^{-1}, y\right)} \\
\gamma_x(a, b) & :=\frac{\omega(a, b, x) \omega\left(x, xax^{-1},  xbx^{-1}\right)}{\omega\left(a, x,  xbx^{-1}\right)} \ .
\end{aligned}\ee
Other structures in $D^\omega (G)$ can also be specified, see~\cite{Gruen:2021afw}. The Drinfeld center $\CZ(\vecgo)$ is also well known to be equivalent to the representation category $\Rep(D^\omega(G))$.

\subsection{Fell Bundle}\label{sec:fellbundle}
A more $C^*$-algebraic formulation of Drinfeld center is given in~\cite{Bhowmick:2018}. The key insight of this work is that the tube algebra of $\vecg$\footnote{In their work they described generic $C^*$ tensor categories graded by $G$, we use only a very naive case of their result.} has a Fell bundle structure. This is also true for the twisted version $\vecgo$, we only have to perform a corresponding twist to the Fell bundle $\Sigma_\omega$. With this regard, the Drinfeld center can be computed from the representation category of the \calg associated to the Fell bundle,
\be\label{eq:centertofellbundle} \CZ(\vecgo) \cong \Rep( C^*(G//_{\rm Ad}G , \Sigma_\omega) )  \ . \ee
In this section we will give a concise introduction to the concept of Fell bundle, for a detailed definition, please refer to Appendix~\ref{app:fellbundle}.

For our purpose, we only consider Fell line bundles. A Fell line bundle over the groupoid $\CG = \gmg$, denoted $p:\CL \to \CG$, contains the following data:
\begin{itemize}
    \item For every $g\in \CG^{(0)}$, each fiber $p^{-1}(g)$ is a one dimensional Banach space (isomorphic to $\bbC$ as a vector space). It is equipped with the complex number multiplication, complex conjugation as an involution map and the complex number norm.
    \item For every morphism $(g,h)\in {\rm Hom}_\CG(g,hgh^{-1})$, the fiber $p^{-1}(g,h)$ gives a correspondence in ${\rm Corr}(\bbC,\bbC)\cong \{\text{all Hilbert spaces} \}$ (see Appendix~\ref{app:bimod} for definition), which we choose to be $\bbC$ in the case of a Fell line bundle.
\end{itemize}
Such a Fell line bundle $\CL$ can be constructed from a $U(1)$-central extension $\widetilde{\CG}$ over $\CG$. A $U(1)$-central extension is a groupoid $\widetilde{\CG}$ with the same set of objects as $\CG$, and the morphisms form a $U(1)$-principal bundle $\pi: \CG^{(1)}\to \CG^{(0)}$. See Definition~\ref{def:centextu1} for precise definition. In certain literatures such extension is also abbreviated as a sequence
\be~\label{eq:u1centext}  1 \rightarrow U(1) \rightarrow \widetilde{\CG} \rightarrow \CG \rightarrow 1 \ . \ee
The Fell line bundle $\CL$ associated with $\widetilde{\CG}$ can be constructed as 
\be\CL = (\widetilde{\CG}\times \bbC)/\{ \forall z\in U(1), (\gamma, c) \sim (\gamma z,z^{-1}c)  \} \ .\ee
In the following, we will use the the $U(1)$-central extension of groupoid and its associated Fell line bundle interchangably. For every Fell line bundle, there is an associated $C^*$-algebra of continuous, compactly supported sections denoted $\Gamma_c(\mathcal{G}, \mathcal{A})$, where the algebra multiplication is given by convolution, and involution given by $f^*(x) = \overline{f(x^{-1})}$.

It has been well-established that the extension~\eqref{eq:u1centext} is classified by the second groupoid cohomology $H^2(\gmg,U(1))\cong \bigoplus_{[g]}H^2_{\rm grp}(C_G(g),U(1))$\footnote{The groupoid cohomology of groupoid $\mfr{X}$ can be constructed as the $\check{\rm C}$ech cohomology of the simplicial space $X_\bullet$, where $X$ is an atlas of $\mfr{X}$ with a representable submersion $X\to \mfr{X}$. The second cohomology class $H^2(\gmg,U(1))$ characterizes $U(1)$-gerbe over $\CG$, which is different from $U(1)$-central extension over $\CG$. However, when $H^1(G,U(1)) = H^2(G,U(1))= 0$, $H^2(\CG,U(1))$ is isomorphic to the isomorphism class of $U(1)$-central extensions.}~\cite{Behrend:2008}, and we will explain how this classification manifest in the \calg associated to the line bundle.  Choose a representative $\nu$ of the class $[\nu]\in H^2(\CG,U(1))$, the corresponding \calg of the Fell bundle admits a twisted multiplication given by the following procedure~\cite{Kumjian:1998umx,Bhowmick:2018}:

\begin{enumerate}
    \item For each $g\in G$, we can take a section $\theta_g : \CG_g\to \Sigma_g$\footnote{The notation $\CG_g$ stands for the set of morphisms starting from the object $g$.}
    \item For each $g\in G$, the cocycle $\nu$ gives a map\footnote{The notation $\CG^g$ stands for the space of morphisms targeting to the object $g$.} 
    \be\ba
    &\nu_g: \CG_g\times \CG^g \to U(1) \\ 
    &\nu_g(x,y) = \theta_{s(y)}(xy)^{-1} \theta_g(x) \theta_{s(y)}(y) \ ,
    \ea\ee
    where the second line holds because the RHS lives in $U(1) \cong U(1) \times \{s(y) \}\subset \Sigma_{ s(y)}$. More explicitly, when restricted to $\CG_g^g$, the cocycle gives
    \be  \nu_g : \CG_g^g\times \CG_g^g \to U(1) \ . \ee
    Given that $\CG_g^g = C_G(g)$, we have
    \be  \nu_g : C_G(g)\times C_G(g) \to U(1)  \ . \ee
    \item The maps $\{\nu_g\}$ satisfy the cocycle identity\footnote{The notation $\CG^g_h := \CG_h\cap \CG^g$.}
    \be\label{eq:cocycle}
    \nu_g\left(x, y\right) \nu_h\left(x y, z\right)=\nu_g\left(x, y z\right) \nu_h\left(y, z\right) \ , \ \forall x \in G_g, y \in G_h^g, z \in G^h
    \ee
    \item The \calg of Fell bundle sections $C^*(\CG,\Sigma)$ has the multiplication given by the $\nu$-twisted convolution
    \be
    (f*g)(x) = \sum_{y\in \CG^{t(x)}} \nu_{s(x)}(y,y^{-1}x) f(y) g(y^{-1}x) 
    \ee
    as well as the involution
    \be
    f^*(x) = \overline{\nu_{s(x)}(x,x^{-1}) f(x^{-1})} \ .
    \ee
\end{enumerate}
The corresponding groupoid $2$-cocycle classifying the Fell bundle $\Sigma_\omega$ on $\gmg$ can be computed from a representative $\omega$ of $[\omega]\in H^3_{\rm grp}(G,U(1))$ by~\cite{Bhowmick:2018}
\be
\nu\left((g_1,s_1), (g_2,s_2)\right)=\omega\left(s_1 g_1 s_1^{-1}, s_1, s_2\right) \omega^{-1}\left(s_1, g_1, s_2\right) \omega\left(s_1, s_2, g_2\right) \ 
\ee
for every $g_{1,2},s_{1,2} \in G \ $ satisfying $ \ g_1 = s_2g_2s_2^{-1}$(i.e., $(g_1,s_1) $ and $ (g_2,s_2)$ are composable morphisms). At the level of $C^*$-algebra, the representation category of $C^*(\gmg,\Sigma_\omega)$ has the same objects as $\CZ(\vecgo)$. With suitable comultiplication and $R$-matrix, we are able to recover the Drinfeld double, and reach the entire modular tensor category structure of $\CZ(\vecgo)$.

\section{Non-Anomalous Lie Group Symmetry}\label{sec:non-anomalous}
A categorical description of Lie group symmetries requires tensor (higher-)categories with infinite simple objects. Various proposals have been presented, including skyscraper sheaf categories~\cite{Freed:2009qp,Jia:2025vrj} and quasi-coherent sheaf categories~\cite{Jia:2025vrj,Stockall:2025ngz}.

Here we present another formulation using the representation categories of infinite dimensional $C^*$-algebras, generalizing the results of finite groups constructed in the previous sections. In this section, we will assume the Lie group $G$ is locally compact Hausdorff (LCH) unless otherwise specified.

\subsection{Construction of Symmetry Category}
Consider the LCH Lie group $G$ with Haar measure $d\Gamma$, we intend to imitate the previous construction by taking the representation category of its function space. However, the $C^*$ property requires the functions to be well behaved under norms, and the topology of $G$ shall be reflected in certain notions of continuity. Thus, we define\footnote{We are using the same notation $C_0(G)$ as in~\cite{Marin-Salvador:2025stc}.}
\be
C_0(G) := \{ f: G\to \bbC | \text{continuous, vanishing at infinity}  \} \ .
\ee
$C_0(G)$ can also be endowed with a Hopf algebra structure~\cite{Etingof:2024abc}, rendering it a Hopf $C^*$-algebra. Specifically, the coproduct is given by
\be
\Delta: C_0(G) \to C_0(G)\otimes C_0(G)\cong C_0(G\times G)\ , \ \Delta(f)(g,h) = f(gh) \ ,
\ee
the counit and antipode are given by
\be\ba
&\varepsilon:C_0(G) \to \bbC \ , \ \varepsilon(f) = f(e)\\
&S:C_0(G) \to C_0(G) \ , \ S(f)(g) = f(g^{-1}) \ .
\ea\ee

With the above Hopf $C^*$-algebra construction, we are able to construct the symmetry category. We define the category ${\rm Hilb}(G)$ with the following data:
\begin{itemize}
    \item The objects are representations of $C_0(G)$. Each Hilbert space $\CH$ of the  representation $\pi : C_0(G) \to \mathscr{B}(\CH)$ necessarily takes the form of a direct integral (See Theorem 1.3.13 of~\cite{Renault:2009})
    \be
    \CH \cong \int_G^\oplus  \CH_g \  d\mu(g) \ ,
    \ee
    where $d\mu$ is a Radon measure\footnote{A Radon measure is a Borel measure that is always finite on compact subsets.} on $G$ (note that this is generically not the Haar measure) and $\CH_g$ is a measurable Hilbert bundle on $G$. Each vector in $\CH$ takes the form
    \be \xi = \int_G  \xi_g \ d\mu(g) \  , \ \xi_g \in \CH_g \ .\ee
    The action of $C_0(G)$ on $\CH$ is given by a pointwise\footnote{In some literatures, such property is denoted ``decomposable''.} multiplication
    \be \big(\pi(f)\cdot \xi\big)_g = f(g)\xi_g \ .  \ee
    \item Every simple object is labeled by a 1-dimensional Hilbert space $\bbC$ over the Dirac measure on a single point, denoted by $(\bbC, \delta_g)$.
    \item Morphisms between two representations are intertwiners.
    \item The monoidal structure is given by a convolution on both the Hilbert bundles and the measure,
        \be \ba
        m: G\times G\to G \ &, \ (g,h)\mapsto gh \\
        d\mu^*\nu &:= m_*(d\mu\times d\nu)\\
        (d\mu, \CH_g) \otimes (d\nu,\CK_g) &= (d\mu^*\nu, (\CH*\CK)_g ) \\
        (\CH*\CK)_g &= \int_G^{\oplus} \ \CH_{h}\otimes \CK_{h'} \ d\lambda_g(h,h')
        \ea \ee
        where $d\lambda_g$ is the conditional measure of product measure $d\mu\times d\nu$ supported on $m^{-1}(g)\subset G\times G$. For simple objects, this reads
        \be (\bbC,\delta_g) \otimes (\bbC,\delta_h) = (\bbC,\delta_{gh}) \ .  \ee
\end{itemize}
The ${\rm Hilb}(G)$ category matches our intuition for a naive generalization of $\vecg$, while providing a mathematically rigorous formulation to the continuous ``direct sum'' of simple objects with the direct integral.

\subsection{Calculating the Drinfeld Center}
In~\cite{hataishi2025categorical,Neshveyev2016}, the following proposition is proved:
\begin{prop}
    The Drinfeld center of unitary representations of a Hopf algebra $\bbG$ is the category of unitary representations of its quantum double. Explicitly,
    \be \CZ(\Rep(\bbG)) = \Rep(D\bbG)  \ . \ee
\end{prop}
In our case, we take $\bbG = C_0(G)$. The quantum double of $C_0(G)$ is a $C^*$-algebra~\cite{Koornwinder:1996uq}
\be
DC_0(G) \cong C_0(G\times G)\ , 
\ee 
with the $*$-algebra structure given by\footnote{We assume unimodularity for group $G$. If unimodularity is not satisfied, the involution shall be modified into $F^*(g,h) = \Delta(h^{-1})\overline{F(h^{-1}gh,h^{-1})}$, where $\Delta:G\to \bbR$ is the modular function (instead of the comultiplication).}
\be\ba
F^*(g,h) &= \overline{F(h^{-1}gh,h^{-1})}\\
(F_1\bullet F_2)(g,h) &= \int_G \  F_1(g,k) F_2(k^{-1}gk , k^{-1}h) \ d\Gamma(k) \ ,
\ea\ee
where $\bullet$ is the algebra multiplication and $d\Gamma$ stands for the Haar measure as before. 
Moreover, Corollary 3.10 of~\cite{Koornwinder:1996uq} reveals that 
the simple objects (irreducible representations) of $\Rep(DC_0(G))$ are classified by $([g],\rho)$, where $[g]$ is a conjugacy class and $\rho$ is an irreducible representation of $C_G(g)$ the centralizer, which aligns with the case of finite groups.

We can make an equivalent formulation with the Fell bundle construction. By taking the Fell bundle $\Sigma$ to be a trivial bundle over the inertia groupoid $\gmg$, the associated \calg is exactly isomorphic to the crossed product (Section 8.1 of~\cite{Kaliszewski:2008})
\be  C^*( \gmg ,\Sigma )  \cong  C_0(G)\rtimes_{\rm Ad} G \ . \ee
This is also exactly the underlying $C^*$-algebra of the Drinfeld double constructed above (See Example 5.13 of~\cite{Roy:2015}).

The Fell bundle over the inertia groupoid provides a smooth pathway to depict the equivariant formulations on Lie groups. We will see how the anomaly information twists the Fell bundle, and how it generates a model for Drinfeld center for the anomalous Lie group symmetry categories.

\section{Anomalous Lie Group Symmetry}\label{sec:anomalous}
We consider the anomalous Lie group symmetry categories specifically in $2$D, where the anomaly is classified by $H^4(BG,\bbZ)$. 
We will first show how the anomaly classification can be transgressed into different cohomology classes, which makes them suitable for different geometric and algebraic models from which we build the anomalous symmetry category and its Drinfeld center. Then we will construct the anomalous symmetry categories from \calg representation, and compute the simple objects of the Drinfeld center through Fell bundle construction.

\subsection{Transgression of Cohomology Classes}

There are various ways to equivalently describe a class $k\in H^4(BG,\bbZ)$. We will use the key mathematical concept of transgression, which was initially developed by Serre, Cartan, Borel et al. to transform the characteristic classes over $BG$ through the fibration of universal principal bundle
\be  G \to EG \to BG \ee
onto cohomology classes on $G$~\cite{CartanSeminar:1948,PierreSerre:1951}. Then Dijkgraaf and Witten connected the WZW term with Chern-Simons theory through transgression
\be\ba \label{eq:ordinarytransgression} \tau: H^4(BG,\bbZ) &\to H^3(G,\bbZ) \\ \text{CS theory} &\to   \text{WZW model} \ea\ee
in their famous work~\cite{Dijkgraaf:1990}. However, our approach to transgression is a bit different. For our purposes, we need certain transgression maps to establish a connection among the ordinary cohomologies $H^4(BG,\bbZ)$ and $H^3(G,\bbZ)$, the equivariant cohomology $H^3_G(G,\bbZ)$ and the groupoid cohomology $H^2(\gmg,U(1))$. Establishing this link enables us to translate the anomaly information into the Fell bundle construction, which forges a handy tool towards the Drinfeld center.

We will briefly introduce the notion of bundle gerbes here, a more rigorous and detailed definition can be found in Appendix~\ref{app:multbdgb}. Consider a manifold $M$, we take a surjective submersion $\varphi:Y\to M$, which can be imagined as a refined cover. Then we take a $U(1)$-bundle $L\to Y\times_M Y$ over the fiber product of $Y$, as well as an isomorphism of $U(1)$-bundles on the triple fiber product $Y\times_M Y\times_MY$, given by
\be  \rho: L_{12} \otimes L_{23} \to  L_{13}\,, \ee
where $L_{i j}$ is the pullback of $L$ to the $(i,j)$-th factor of $Y^{[3]}$. There is also a coherence relation of $\rho$ on the quadraple fiber product.

When we consider a bundle gerbe over a group $G$, the multiplication structure $m:G\times G \to G$ becomes important. Bundle gerbes incorporated with the multiplication is called \emph{multiplicative bundle gerbe}. Explicitly, suppose $\mathscr{G}$ is a bundle gerbe over $M = G$, the group multiplication gives an equivalence of bundle gerbes over $G\times G$,
\be \mathscr{M}: p_1^* \mathscr{G} \otimes  p_2^* \mathscr{G}  \to m^*\mathscr{G}\,, \ee
This equivalence invites an isomorphism $\alpha$ over $G\times G\times G$ due to the associativity of group multiplication, and $\alpha$ shall give a coherent identity over $G^{\times 4}$. We will denote a multiplicative bundle gerbe by its data $(\mathscr{G},\mathscr{M},\alpha)$. The associativity given by $\alpha$ is exactly analogous to the cohomology class $[\omega]\in H^3_{\rm grp}(G,U(1))$ classifying finite group anomaly in $2$D.

Murray proved that the equivalence class of bundle gerbes over $G$ is classified by an integral class $[\kappa]\in H^3(G,\bbZ)$~\cite{Murray:1996bdg,Murray:2000sit}, also know as the \emph{Dixmier-Douady} class.  For multiplicative bundle gerbes, a crucial input is provided by Proposition 5.2 of~\cite{Carey:2004xt} stating that:
\begin{prop}[Classification of multiplicative bundle gerbe~\cite{Carey:2004xt}]
    Let $G$ be a compact, connected Lie group. Then there is an isomorphism between $H^4(B G , \mathbb{Z})$ and the space of isomorphism classes of multiplicative bundle gerbes $(\mathscr{G},\mathscr{M},\alpha)$ on $G$.
\end{prop}
In other words, multiplicative bundle gerbes are just bundle gerbes whose Dixmier-Douady classes lie in the image of $\tau$ in \eqref{eq:ordinarytransgression}. Thereafter, in our previous work~\cite{Jia:2025vrj}, we gave a geometric construction on how the multiplicative bundle gerbe could result in a twist in the categories of skyscraper sheaves on Lie groups.

Bundle gerbes are crucial in the geometric construction, and the geometric intuition of quantum double would lead to $G$-action by conjugation, rendering the bundle gerbe $G$-equivariant. For such bundle gerbes, the classification become the equivariant cohomology
\be  H^3_G(G,\bbZ) =  H^3(EG\times_G G,\bbZ) \ . \ee
The relation with $H^3(G,\bbZ)$ is given by Proposition 3.1 of~\cite{Carey:2008}, which states that they are isomorphic,
\be  H^3_G(G,\bbZ)  \cong H^3(G,\bbZ)  \ee
for any connected, compact, simply-connected simple Lie group $G$. The equivariant cohomology again equals the groupoid cohomology~\cite{BarbosaTorres:2021}\footnote{Their work stated such equation based on the cohomology theory of differential stacks, not on groupoids. However, the stack cohomology and groupoid cohomology are isomorphic when we consider the groupoid as an atlas of the stack. See Definition 3.2, Theorem 4.3 and Theorem 4.7 of~\cite{BarbosaTorres:2021}.}
\be  H^3_G(G,\bbZ) =  H^3(G//_{\rm Ad}G,\bbZ)  \ , \ee
which is also isomorphic to the classification of groupoid extension~\cite{Tu:2003ph}
\be H^3(G//_{\rm Ad}G,\bbZ) \cong H^2(G//_{\rm Ad}G,U(1)) \ ,  \ee
which is exactly the classification of Fell line bundles over $\gmg$ discussed in Section~\ref{sec:fellbundle}. 
To summarize, the datum of $2$D anomaly $k\in H^4(BG,\bbZ)$ can be equivalently understood as the construction of a Fell bundle $\Sigma_k$ over the inertia groupoid $\gmg$.

More generically, we do not require the transgression map $\tau:H^4(BG,\bbZ)\to H^3_G(G,\bbZ)$ to be an isomorphism, but rather an injective homomorphism. We have proven in Appendix~\ref{app:injectivetransgression} that the transgression is indeed an injective homomorphism for compact connected Lie groups. Furthermore, we can also discuss non-compact Lie groups. Consider a non-compact reductive algebraic Lie group $G$ (over $\bbR$ or $\bbC$) and its maximal compact subgroup $K$, then the inclusion map $i:K\hookrightarrow G$ is a strong deformation retraction and induces the homotopy equivalences (see Theorem 3.1 of~\cite{Adem:2013pjx})
\be BG\simeq BK \ , \ EG\simeq EK \ee
and also $ EG\times _G G \simeq EK\times_K K $. Thus the cohomology classes are isomorphic
\be H^4(BG,\bbZ) \cong H^4(BK,\bbZ) \ , \ H^3_G(G,\bbZ) \cong H^3_K(K,\bbZ) \ .  \ee
Therefore the transgression map is also valid for non-compact connected reductive algebraic Lie groups (over $\bbR$ or $\bbC$).

\subsection{Construction of Symmetry Category}
For a given anomaly class $k\in H^4(BG,\bbZ)$, there is a corresponding multiplicative bundle gerbe $(\mathscr{G},\mathscr{M},\alpha)$, we can construct the twisted symmetry category ${\rm Hilb}^k(G)$ following the continuous tensor category procedure (See Example 3.12 in~\cite{Marin-Salvador:2025stc}). Specifically, we should first construct a twisted version of $C_0(G)$, consider its representation category then incorporate the twist information in the monoidal structure. In~\cite{Marin-Salvador:2025stc}, they constructed the category ${\rm Hilb}^\omega(G)$ with the twist $\omega\in H^3_{\rm SM}(BG,U(1))$ in the Segal-Mitchison cohomology. We will show that this is equivalent to $H^4(BG,\bbZ)$, and thus equivalent to the multiplicative bundle gerbes.

For the anomaly classes in $H^4(BG,\bbZ)$, we cannot naively claim that the Bockstein homomorphism
\be
\beta : H^3(BG,U(1)) \to H^4(BG,\bbZ)
\ee
is an isomorphism, because $H^\bullet(BG,\bbR)$ are not always trivial (for finite groups they are). However, this can be true for differentiable group cohomology. In~\cite{Brylinski:2000dcg,Wagemann:2015ctg}, and also explicitly summarized in~\cite{Liu:2026atf}, we have the following results:
\be H^4(BG,\bbZ) \cong H^3_{\rm SM}(BG,U(1)) \cong H^3(BG,\underline{U(1)}) \ .  \ee

The construction of twisted symmetry category ${\rm Hilb}^k(G)$ proposed as Example 3.12 of~\cite{Marin-Salvador:2025stc} is indeed valid, but we will adopt an equivalent formulation with multiplicative bundle gerbes to emphasize the geometric perspective (we give definition and classification of multiplicative bundle gerbe in Appendix~\ref{app:multbdgb}). Briefly speaking, we take the following steps:
\begin{enumerate}
    \item The anomaly $k\in H^4(BG,\bbZ)$ gives a multiplicative bundle gerbe $(\mathscr{G},\mathscr{M},\alpha)$.
    \item The bundle gerbe $\mathscr{G}$ includes the information of a line bundle $L\to Y^{[2]}$ over the surjective submersion $\pi : Y\to G$. The function $\rho: Y^{[3]} \to U(1)$  is the $\check{\rm C}$ech 2-cocycle that dictates the trivialization of the gerbe multiplication. The $\rho$ function gives a twisted convolution on the vector space $C_c(Y^{[2]})$
    \be  (f_1*f_2)(x,z) = \int_{(x,y,z)\in Y^{[3]}} f_1(x,y) f_2(y,z) \overline{\rho(x,y,z)} \ dy \ .  \ee
    The involution is given by
    \be f(x,y )=  \overline{f(y,x)} \ .  \ee
    This vector space can be completed to be a $C^*$-algebra $C_0(G,\rho)$\footnote{The nomenclature for $C_0(G,\rho)$ is intuitive but misleading. We gave it such name because we build the twisted symmetry category from its representation, but it is drastically different from $C_0(G)$ or $C^\omega(G)$. The \calg $C_0(G,\rho)$ is the normed function space of points on the fiber product $Y\times_G Y$. In the $\rho = 0$ scenario, it is still larger than $C_0(G)$. However, $C_0(G,\rho = 0) \cong C_0(G)\otimes \mathscr{B}(\CH)$ and $C_0(G)$ are Morita equivalent (see II.7.6.10 of~\cite{Blackadar:2006}) in the sense of operator algebras, which means they share equivalent representations and correspondences.}.
    \item The set of objects of $\hilb^k(G)$ is the set of representations $\Rep(C_0(G,\rho))$. The simple objects (irreducible representations) are exactly labeled by single points on the group manifold, which we denote as $\delta_g$.
    \item The monoidal structure should be induced by a map
    \be  \Rep(C_0(G,\rho))\otimes \Rep(C_0(G,\rho)) \cong  \Rep(C_0(G,\rho)\otimes C_0(G,\rho)) \to \Rep(C_0(G,\rho)) \ ,  \ee
    which is implemented by a $C_0(G,\rho)-C_0(G,\rho)\otimes C_0(G,\rho)$ correspondence (for definitions on correspondence of $C^*$-algebras, please refer to Appendix~\ref{app:bimod}). This correspondence basically encodes the information of $\mathscr{M}:p_1^*\mathscr{G}\otimes p^*_2\mathscr{G}\to m^*\mathscr{G}$. This information trickles down the representation category through the following procedure. Notice that there is a pullback hidden in $C_0(G,\rho)\otimes C_0(G,\rho)$
    \be  \ba
    &C_0(G,\rho) \otimes C_0(G,\rho) \cong C_0(G\times G,(p_1+p_2)^*\rho) \\
    &Y^{[3]}\times Y^{[3]} \stackrel{\rho \times \rho }{\longrightarrow} U(1) \times U(1) \stackrel{m}{\longrightarrow} U(1) 
    \ea \ee
    which transforms $\rho$ on $\mathscr{G}$ into $p_1^*\rho+p_2^*\rho$ on $p_1^*\mathscr{G}\otimes p^*_2\mathscr{G}$. The equivalence of bundle gerbes $\mathscr{M}$ dictates that the two pullbacks differs only by a coboundary
    \be p_1^*\rho + p_2^*\rho = m^*\rho + d \upsilon  \ . \ee
    This provides an invertible $C^*$-correspondence \be\mathcal{T}_{\mathscr{M}} \in {\rm Corr}\Big(C_0(G\times G,m^*\rho),C_0(G\times G,p_1^*\rho + p_2^*\rho)\Big) \ . \ee
    By taking relative tensor (Rieffel tensor) with 
    \be C_0(G\times G,m^*\rho)\in {\rm Corr}\big( C_0(G,\rho) ,  C_0(G\times G,m^*\rho) \big) \ , \ee
    we get the correspondence we want
    \be
    \mathcal{C}_{\mathscr{M}} := C_0(G\times G,m^*\rho) \otimes_{C_0(G\times G,m^*\rho)} \CT_{\mathscr{M}} \in {\rm Corr}\Big(C_0(G,\rho) , C_0(G,\rho)\otimes C_0(G,\rho)\Big) \ .
    \ee
    For any $R_1, R_2 \in \Rep(C_0(G,\rho))$, the monoidal structure of the category gives
    \be \ba \otimes_{\hilb^k(G)} : R_1 \otimes R_2  \mapsto   \CC_{\mathscr{M}} \otimes_{C_0(G,\rho) \otimes C_0(G,\rho)} (R_1 \otimes R_2) \in &{\rm Corr}(C_0(G,\rho), \bbC) \\ &= \Rep(C_0(G,\rho))  \ . \ea\ee
    \item Associativity of tensor product is given by $\alpha$ in the multiplicative bundle gerbe
    \be
    \begin{tikzcd}
\mathscr{G}_1\otimes \mathscr{G}_2\otimes \mathscr{G}_3 \arrow[d, "{{\rm id }\otimes \mathscr{M}_{2,3}}"'] \arrow[rr, "{\mathscr{M}_{1,2}\otimes {\rm id}}"] &  & \mathscr{G}_{12}\otimes \mathscr{G}_3 \arrow[d, "{\mathscr{M}_{12,3}}"] \arrow[lld, "{\alpha_{1,2,3}}", Rightarrow] \\
\mathscr{G}_{1}\otimes \mathscr{G}_{23} \arrow[rr, "{\mathscr{M}_{1,23}}"']                                                                          &  & \mathscr{G}_{123}                                                                                                      
\end{tikzcd}\ee
specifying the associativity of $\mathscr{M}$. For the \calg scenario, it translates to be a unitary intertwiner (for simplicity we denote $A = C_0(G,\rho)$)
\be\ba
\beta : \CC_\mathscr{M}\otimes_{A\otimes A}\Big\{ &\Big[ \CC_\mathscr{M}\otimes_{A\otimes A}( R_1 \otimes R_2) \Big] \otimes R_3 \Big\} \\ &\to \CC_\mathscr{M}\otimes_{A\otimes A} \Big\{R_1\otimes\Big[\CC_\mathscr{M}\otimes_{A\otimes A}(R_2\otimes R_3)\Big]\Big\} \ .
\ea\ee
\end{enumerate}

A simple object $\delta_g\in {\rm Obj}(\hilb^k(G))$ is an irreducible representation $\pi_g : C_0(G,\rho) \to \mathscr{B}( L^2(Y_g) )$\footnote{We will refer to the simple object as $\delta_g$ when we discuss categorical data, $\pi_g$ when we discuss the representation of $C_0(G,\rho)$.}. For the fiber on $g\in G$, the Čech cocycle data $\rho$ can be represented by a function $b_g: Y^{[2]}\to U(1)$ s.t.
\be b_g(x,y) b_g(y,z) = \overline{\rho(x,y,z)}  b_g(x,z) \ . \ee
The representation $\pi_g$ is given by
\be \big(\pi_g(f)\psi\big)(x) = \int_{Y_g} f(x,y) b_g(x,y) \psi(y) \, dy  \ , \ x\in Y_g \ . \ee
Then when performing composition we can see $b_g$ reproduces $\overline{\rho}$,
\be
\big(\pi_g(f_1*f_2)\psi\big)(x) = \Big(\pi_g(f_1)\circ \big(\pi_g(f_2) \psi\big) \Big)(x) \ .
\ee
The monoidal structure is induced by tensoring with $\CC_\mathscr{M}$. The data $\mathscr{M}$ gives the equivalence between bundle gerbes $p_1^*\mathscr{G}\otimes p^*_2\mathscr{G}$ and $m^*\mathscr{G}$ over $G\times G$. On fibers, this gives $Y_g\times Y_h\to Y_{gh}$, and one can verify
\be  \pi_g \otimes_{\hilb^k(G)}\pi_h \cong \pi_{gh} \ . \ee
Then the associator
\be \beta_{g,h,k}: \big(\pi_g \otimes_{\hilb^k(G)} \pi_h\big) \otimes_{\hilb^k(G)} \pi_k \stackrel{\sim}{\longrightarrow} \pi_g \otimes_{\hilb^k(G)} \big(\pi_h \otimes_{\hilb^k(G)} \pi_k \big) \ee
must be a $U(1)$ factor due to Schur's lemma.

\subsection{Intuition of Drinfeld Center}\label{sec:intuitioncenter}
In order to develop an intuition for the Drinfeld center, we observe that $H^4(BG,\mathbb Z)$ classifies a principal $B^2U(1)$-2-bundle over $BG$, namely, we have:
\begin{equation}\label{eq:Point_Level_Phys_Proof}
    k: BG \rightarrow B^3U(1)
\end{equation}
as a map from $BG$ to $B^3U(1)$. Apply a looping to $k$, we have (see Lemma 17.13 and 17.14 of~\cite{Behrend:2011ui}, and~\cite{Bunke:2006ipq}):
\begin{equation}\label{eq:Looping_Phys_Proof}
    L(k) : LBG \cong [G/G] \rightarrow LB^3U(1) \cong B^3U(1)\times B^2U(1),
\end{equation}
where $[G/G]$ is the quotient stack with $G$ acting on itself by conjugation.
The transgression is then constructed by projecting to the second factor as~\cite{fiorenza2015higher}:
\begin{equation}\label{eq:Loop_Level_Phys_Proof}
    \tau(k) = p_2\circ L(k): [G/G] \rightarrow B^2U(1)\,.
\end{equation}
Thus, $\tau(k)$ characterizes a principal $BU(1)$-bundle over $[G/G]$, or equivalently, an $S^1$-bundle-gerbe over $[G/G]$~\cite{Nikolaus:2011ag}. This then leads to a $U(1)$-central extension~\cite{behrend2003equivariant, Behrend:2008}:
\begin{equation}\label{eq:Central_Extension_of_G//G}
    1 \rightarrow U(1) \xrightarrow{\tau(k)} \Sigma_k \rightarrow G//_{\rm Ad} G \rightarrow 1
\end{equation}
which is exactly the sought-after loop-level data. Furthermore, since $H^4(BG, \mathbb{Z})$ can alternatively be viewed as classifying multiplicative bundle gerbes over $G$, which in turn is equivalent to $BU(1)$-extensions of $G$~\cite{waldorf2012construction}, the above procedure effectively constructs a path from\footnote{In many literatures, such extension by $BU(1)$ is denoted a string 2-group.}:
\begin{equation}\label{eq:Central_Extension_of_G}
    1 \rightarrow BU(1) \xrightarrow{k} \mathcal G_k \rightarrow G \rightarrow 1
\end{equation}
to~\eqref{eq:Central_Extension_of_G//G}. It is then clear that intuition leading to the looping operation in obtaining~\eqref{eq:Looping_Phys_Proof} can alternatively be justified by the construction of Hochschild homology from factorization homology~\cite{Ben-Zvi:2008vtm, ben2012morita, haine2019nonabelian} (see in particular Corollary 4.6 of~\cite{ayala2015factorization} applied to $B = S^1$ and $X = BG$):
\begin{equation}
    \int_{S^1} G \cong \int_{S^1} \Omega BG \cong LBG\,.
\end{equation}
The reader can further consult~\cite{Amabel:2019yrn, Costello:2021jvx} for basics of factorization homology.

In the framework of the Cobordism Hypothesis, one can view~\eqref{eq:Point_Level_Phys_Proof} as the point-level data that the corresponding TQFT assigns to a point, and~\eqref{eq:Loop_Level_Phys_Proof} as the loop-level data the same TQFT assigns to a loop. The latter is obtained from the former naturally either by looping in going from~\eqref{eq:Point_Level_Phys_Proof} to~\eqref{eq:Loop_Level_Phys_Proof}, or alternatively by $\int_{S^1}$-operation in factorization homology in going from~\eqref{eq:Central_Extension_of_G} to~\eqref{eq:Central_Extension_of_G//G}. Furthermore, in the spirit of~\cite{Landsman:1999, Hawkins:2008, landsman2012mathematical}, one can quantize the theory by linearizing $\Sigma_k$ over $G//_{\rm Ad}G$ to be $C^*(G//_{\rm Ad}G, \Sigma_k)$, whose quantum sectors are described by $\Rep(C^*(G//_{\rm Ad}G, \Sigma_k))$. Thus, by the Cobordism Hypothesis and the standard quantization procedure, one should intuitively expect that $\Rep(C^*(G//_{\rm Ad}G, \Sigma_k))$ is the Drinfeld center of $\Rep^k(C_0(G))$. This constitutes an intuition behind the formalism. 

\subsection{Calculating the Drinfeld Center}
As explained in Section~\ref{sec:intuitioncenter}, our calculation of Drinfeld center $\CZ(\hilb^k(G))$ is transformed to the calculation of the representation category of the \calg associated to the Fell bundle $C^*(\gmg,\Sigma_k)$, which captures the twisted equivariant structure.

Fortunately, the representations of $C^*(\gmg,\Sigma_k)$ has been described in Theorem 2.1 of~\cite{Ionescu:2013},
\begin{theorem}[Irreducible representations of Fell bundle]
    Suppose that $p: \mathscr{B} \rightarrow \CG$ is a separable Fell bundle over a locally compact groupoid with a Haar system. Let $u \in \CG^{(0)}$ and let $\CG_u^u:={\rm Hom}_{\CG}(u,u) $ be the stability group at $u$. Suppose that $L$ is an irreducible representation of $C^*(\CG_u^u, \mathscr{B})$. Then $\operatorname{Ind}_{\CG(u)}^G L$ is an irreducible representation of $C^*(\CG, \mathscr{B})$.
\end{theorem}
In our case of compact Lie group $G$, the inertia groupoid $\CG =\gmg$ is locally compact with a Haar system induced from the Haar measure, and the Fell line bundle is also separable. For $g\in G = \CG^{(0)}$, the stability group $\CG_g^g$ is nothing more than the centralizer $C_G(g)$. On $C_G(g) = \CG_g^g$, the cocycle $\tau(k)$ gives a map (c.f. Section~\ref{sec:fellbundle})
\be\label{eq:transgrprojrep} \tau(k)_g:  C_G(g) \times C_G(g) \to U(1)  \ee
with cocycle identity given by~\eqref{eq:cocycle}. This means irreducible representations on $C^*(\CG_g^g,\Sigma_k)$ are the irreducible projective representations of $C_G(g)$. Thus we have the following corollary
\begin{corollary}
    Suppose $G$ is a compact connected Lie group (or a non-compact connected reductive algebraic group over $\bbR$ or $\bbC$), $\Sigma_k$ is a Fell line bundle over $\CG = G//_{\rm Ad} G$ classified by $\tau(k)\in H^2(\gmg,U(1))$. For any $g\in G$, if $L$ is an irreducible representation of $C^*(\CG_g^g ,\Sigma_k)$, then the induced representation ${\rm Ind}_{\CG^g_g}^{\CG}L$ is an irreducible representation of $C^*(\CG ,\Sigma_k)$.
\end{corollary}
One can verify the equivalence of induced representations when we choose different elements in the same conjugacy class. When the groupoid $\gmg$ is amenable (iff the group $G$ is amenable), all irreducible representations are induced representations from $\Rep(C^*(\CG_g^g ,\Sigma_k))$~\cite{Ionescu:2009}. The amenability condition is always satisfied for compact groups.

With the above remarks, we are able to state the following proposition:
\begin{prop}
    Suppose $G$ is a compact connected Lie group (or a non-compact connected amenable reductive algebraic group over $\bbR$ or $\bbC$), $\Sigma_k$ is a Fell line bundle over $\CG = G//_{\rm Ad} G$ classified by $\tau(k)\in H^2(\gmg,U(1))$. The irreducible representations of $C^*(\CG,\Sigma_k)$ are classified by
    \be ([g], R )  \ee
    where $[g]$ is a conjugacy class of $G$ and $R$ is an irreducible projective representation of the centralizer $C_G(g)$ twisted by $\tau(k)_g$.
\end{prop}
and with that we state:
\begin{prop}\label{prop:centerclassify}
    The simple objects of $\CZ(\hilb^k(G))$ are classified by
    \be ([g], R )  \ee
    where $[g]$ is a conjugacy class of $G$ and $R$ is an irreducible projective representation of the centralizer $C_G(g)$ twisted by $\tau(k)_g$.
\end{prop}
This result matches the direct generalization from finite group scenarios.

We would like to discuss the continuous analog of modular structure at the level of simple objects. Since there are fatal analytical obstructions for non-Abelian groups and non-compact groups, we only discuss compact connected Abelian groups. Suppose $G$ is a compact connected Abelian Lie group, every $g\in G$ is a conjugacy class and $C_G(g) = G$. In this case the simple objects in $\CZ(\hilb^k(G))$ are labeled by $\big(g\in G,R\in {\rm IrRep}^{\tau(k)_g}(G)\big)$. Then the analog of modular matrices given by~\cite{Coste:2000tq}
\be\ba S_{(g,R_1),(h,R_2)} &= \left(\frac{\tau(k)_g(h,e)\tau(k)_h(g,e)}{\tau(k)_g(e,h)\tau(k)_h(e,g)}\right) \chi_{R_1}(h) \chi_{R_2}(g)  \ , \\
T_{(g,R_1),(h,R_2)} &= \delta_{g,h}\delta_{R_1,R_2} \frac{\chi^*_{R_1}(g)}{\chi_{R_1}(e)} \ .
\ea\ee
The continuous analog of modular matrices should be more properly considered as a distribution kernel.
\begin{exm}[$U(1)\times U(1)$]
    Suppose we have $G= U(1)\times U(1) $ with anomaly $\omega\in H^3(B(U(1)_m\times U(1)_w),\underline{U(1)})$ given by~\cite{Jia:2025uun} ($[a+b]:= a+b \ {\rm mod} \ 2\pi$)
    \be
\omega((e^{i\alpha_1},e^{i\beta_1}),(e^{i\alpha_2},e^{i\beta_2}),(e^{i\alpha_3},e^{i\beta_3}))=\exp\left(-\frac{i}{2\pi}\alpha_1(\beta_2+\beta_3-[\beta_2+\beta_3])\right)\,.
\ee
Thus we have \be \label{eq:projdata}\tau(k)_{(e^{i\alpha_1},e^{i\beta_1})}\big((e^{i\alpha_2},e^{i\beta_2}),(e^{i\alpha_3},e^{i\beta_3})\big) = \exp\left(-\frac{i}{2\pi}\alpha_1(\beta_2+\beta_3-[\beta_2+\beta_3])\right)\,,\ee
which is the same to $\omega$. It is then easy to deduce 
\be\tau(k)_g(h,e) = \tau(k)_g(e,h) = 1\ , \ \forall g,h\in G \ . \ee
The irreducible projective representation twisted by~\eqref{eq:projdata} is
\be \pi_{(m_1,m_2)}\big( (e^{i\alpha},e^{i\beta}) \big) = \exp\left(  i \left(m_1\alpha + m_2\beta-\frac{\alpha_1\beta}{2\pi}\right) \right)\,. \ee
Thus we have
\be\ba
\label{eq:STMats}
S_{(e^{i\theta_1},e^{i\theta_2},m_1,m_2),(e^{i\tilde{\theta}_1},e^{i\tilde{\theta}_2},\tilde{m}_1,\tilde{m}_2)}&=\exp\left(i(\theta_1\tilde{m}_1+\tilde{\theta}_1 m_1+\theta_2\tilde{m}_2+\tilde{\theta}_2 m_2-\frac{1}{2\pi}(\theta_1\tilde{\theta}_2+\theta_2\tilde{\theta}_1))\right) \\
T_{(e^{i\theta_1},e^{i\theta_2},m_1,m_2),(e^{i\tilde{\theta}_1},e^{i\tilde{\theta}_2},\tilde{m}_1,\tilde{m}_2)} &= \delta_{\theta_1,\tilde{\theta}_1}\delta_{\theta_2,\tilde{\theta}_2} \delta_{m_1,\tilde{m}_1}\delta_{m_2,\tilde{m}_2}\exp\left(  -i \left(m_1\theta_1 + m_2\theta_2-\frac{\theta_1\theta_2}{2\pi}\right) \right) \ .
\ea\ee
\end{exm}

\section{Physical Examples}\label{sec:physicalexm}

We would like to discuss the flat gauging of continuous Lie group symmetries in physical systems. The central example would be the $2$D compact real scalar $\phi$ with radius $R$, which has the $(U(1)_m\times U(1)_w)\rtimes\mb{Z}_2$ 0-form global symmetry at a generic radius $R\neq 1$, and the enhanced $(SU(2)_L\times SU(2)_R)/\mb{Z}_2$ 0-form global symmetry at the self-dual radius $R=1$. 

\subsection{$U(1)$ and $\mb{R}$}

Let us consider the $2$D compact scalar at the generic radius $R\neq 1$. The $U(1)_m\times U(1)_w$ 0-form global symmetry generated by the currents
\be
j^{(m)}=*d\phi\ ,\ j^{(w)}=d\phi
\ee
that couple to the background gauge fields $A^{(m)}$ and $A^{(w)}$, and has a mixed 't Hooft anomaly
\be
\label{U1U1-anomaly}
I=\frac{1}{2\pi}A^{(m)}\wedge dA^{(w)}\,.
\ee
If one gauge the non-anomalous finite subgroups $\mb{Z}_p\subset U(1)_m$ and $\mb{Z}_q\subset U(1)_w$, gcd$(p,q)=1$, it is known that the radius of the compact scalar is transformed to $R=q/p$. 

Now we consider the flat gauging of $U(1)_w$, after which the radius $R\rightarrow\infty$. Hence the gauged theory is a non-compact real scalar $\phi$, with an $\mb{R}$ 0-form global shifting symmetry $\phi\rightarrow\phi+a$ $(a\in\mb{R})$.

While the discussion of flat gauging for the compact boson was carried out in the BF theory language in \cite{Argurio:2024ewp}, we would like to interpret this physical story in our categorical language.

We start with the SymTFT of $U(1)_m\times U(1)_w$ with a twist $\omega\in H^3(B(U(1)_m\times U(1)_w),\underline{U(1)})$ corresponding to the 't Hooft anomaly~\eqref{U1U1-anomaly}. In our categorical language, it is $\CZ(\hilb^\omega(U(1)_m\times U(1)_w))=\Rep(C^*(\gmg, \Sigma_\omega))$, where $G = U(1)_m\times U(1)_w$ and $\Sigma_\omega$ is the Fell bundle classified by $\omega$. As in~\cite{Jia:2025vrj}, the representations are labeled by $(e^{i\theta_1},e^{i\theta_2},m_1,m_2)$. $e^{i\theta_1}$ and $e^{i\theta_2}$ are conjugacy classes of $U(1)_m$ and $U(1)_w$, and $m_1,m_2\in\mb{Z}$ label the projective representations of $U(1)_m\times U(1)_w$.

The explicit representative of the group cohomology 3-cocycle $\omega$ can be found in~\cite{Jia:2025uun}:
\be
\omega((e^{i\alpha_1},e^{i\beta_1}),(e^{i\alpha_2},e^{i\beta_2}),(e^{i\alpha_3},e^{i\beta_3}))=\exp\left(-\frac{i}{2\pi}\alpha_1(\beta_2+\beta_3-[\beta_2+\beta_3])\right)\,.
\ee

The two anyon line operators in the 2+1d SymTFT $L_{(e^{i\theta_1},e^{i\theta_2},m_1,m_2)}$ and $L_{(e^{i\tilde{\theta}_1},e^{i\tilde{\theta}_2},\tilde{m}_1,\tilde{m}_2)}$, have the braiding correlation function (c.f.~\eqref{eq:STMats})
\be
\label{theta-linking}
\langle L_{(e^{i\theta_1},e^{i\theta_2},m_1,m_2)}L_{(e^{i\tilde{\theta}_1},e^{i\tilde{\theta}_2},\tilde{m}_1,\tilde{m}_2)}\rangle=\exp\left(i(\theta_1\tilde{m}_1+\tilde{\theta}_1 m_1+\theta_2\tilde{m}_2+\tilde{\theta}_2 m_2-\frac{1}{2\pi}(\theta_1\tilde{\theta}_2+\theta_2\tilde{\theta}_1))\right)\,.
\ee
Thus we can observe the identification between anyon lines with different labeling $(\theta_1,\theta_2,m_1,m_2)\sim (\theta_1+2\pi,\theta_2,m_1,m_2+1)$, and $(\theta_1,\theta_2,m_1,m_2)\sim (\theta_1,\theta_2+2\pi,m_1+1,m_2)$. For this reason, the anyon lines in the SymTFT are equivalently described by a set of parameters $(x,y)\in\mb{R}\times\mb{R}$:
\be
x=2\pi m_2-\theta_1\ ,\ y=2\pi m_1-\theta_2\,.
\ee
The braiding correlation function is
\be
\label{xy-linking}
\langle L_{(x,y)}L_{(\tilde{x},\tilde{y})}\rangle=\exp\left(-\frac{i}{2\pi}(x\tilde{y}+y\tilde{x})\right)\,.
\ee
It is straightforward to see the equivalence between~\eqref{theta-linking} and~\eqref{xy-linking}.

From~\eqref{xy-linking}, we see that $L_{(x,y)}$ labels the anyon lines in the SymTFT $\CZ(\hilb(\bbR))=\Rep(DC_0(\mb{R}))$, i.e. the representation category of the untwisted Drinfeld double of $C_0(\mb{R})$. When we choose the set of condensed operators in $\{L_{(x,y)}\}$ with mutually trivial linking correlation function, we can condense all $\{L_{(x,0)}\}$ or all $\{L_{(0,y)}\}$. They correspond to the boundary $2$D theory with a $\mb{R}$ global symmetry, i.e. the non-compact scalar obtained by the flat gauging of either $U(1)_m$ or $U(1)_w$ in the compact scalar.

On the other hand, we can also condense the sublattice $\mb{Z}^2\subset \mb{R}^2$ generated by $(x_1,y_1)$ and $(x_2,y_2)$, satisfying
\be
2x_1 y_1\ ,\ 2x_2 y_2\ ,\ x_1 y_2+y_1 x_2\in (2\pi)^2\mb{Z}\,.
\ee
Thus the sublattices are always generated by $(2\pi a,0)$ and $(0,2\pi/a)$, $a\in\mb{R}_+$.

The resulting theory is a compact scalar with $(U(1)_m\times U(1)_w)\rtimes \mb{Z}_2$ 0-form global symmetry, with in general an irrational radius $R=a$. 

\subsection{$SU(2)$ and $SO(3)$}

Now we consider the compact boson at self-dual radius $R=1$, with $(SU(2)_L\times SU(2)_R)/\mb{Z}_2\cong SO(4)$ 0-form global symmetry.\footnote{In fact, the chiral algebra at the self-dual radius is $\widehat{su(2)}_1{}_L \times \widehat{su(2)}_{-1}{}_R$, if we just take the simply connected group, the element $(-1,-1)\in SU(2)_L\times SU(2)_R$ would act trivially, hence the global structure of the symmetry groups is $(SU(2)_L\times SU(2)_R)/\mb{Z}_2$.} The two $SU(2)$ factors have a mixed anomaly
\be
\label{SU(2)-anomaly}
I=\frac{1}{2\pi}\left(\mathrm{CS}_3(A_L)-\mathrm{CS}_3(A_R)\right)\,.
\ee
The symmetry category in our language is $\hilb^\omega(SO(4))$, where $\omega\in H^3(BSO(4),\underline{U(1)})$ corresponding to the anomaly in~\eqref{SU(2)-anomaly}.

Nonetheless, one can flat-gauge the non-anomalous diagonal subgroup $SU(2)/\mb{Z}_2\cong SO(3)$. Flat gauging of this diagonal $SO(3)$ was studied in~\cite{Gaberdiel:2011aa} under the name of \emph{continuous orbifold}. The full partition function of the flat-gauged theory coincides with one half of the partition function of a non-compact boson. At first sight, one might therefore suspect that the resulting theory is simply the decompactification limit of the free boson on $S^1/\mathbb Z_2$, namely the infinite-radius endpoint of the orbifold branch. However, the authors showed that the gauged theory should instead be viewed as a new nontrivial non-rational $c=1$ CFT, identified with the Runkel-Watts (RW) model~\cite{Runkel:2001ng}, which arises as the $c\to 1$ limit of Virasoro minimal models.

For our purposes, the key point is that the relation between the continuous-orbifold construction and the minimal-model limit provides strong evidence that this wider class of continuous-orbifold CFTs is genuinely well defined. From this perspective, flat gauging may serve as a useful tool for probing the CFT landscape. We also recommend that interested readers consult \cite{Roggenkamp:2003qp} and the dissertation \cite{Restuccia:2013tba} for further related discussions.

After gauging the diagonal $SO(3)$ global symmetry, the new theory contains a remaining 0-form double set symmetry $SO(3)\backslash SO(4)/SO(3)$, as well as a dual Rep$(SO(3))$ non-invertible 0-form symmetry. Due to the mixed anomaly (\ref{SU(2)-anomaly}), these two factors are not a direct product. The resulting categorical symmetry $\mc{C}_{RW}$ is the category of $SO(3)$-bimodules of the original symmetry category (abusing the similar notation as in the case of finite symmetries~\cite{Bhardwaj:2017xup,Tachikawa:2017gyf}):
\be
\mc{C}_{RW}=\mathrm{Bimod}_{SO(3)}\hilb^\omega(SO(4))\,.
\ee

\section{Conclusion and Outlook}\label{sec:outlook}
In this work, we constructed the symmetry category for anomalous compact Lie group symmetries and calculated the Drinfeld center, enabling a convenient physical analysis of flat gauging using continuous SymTFT.

A promising application in condensed matter physics is to consider gapped and gapless boundary conditions of continuous SymTFT (``the categorical Landau paradigm''~\cite{Bhardwaj:2023fca,Bhardwaj:2023bbf}) and classify phases with Lie group symmetries. Different from finite symmetries, here we will expect gaplessness resulted from Goldstone bosons when we break Lie group symmetries. We should also expect phases with continuous non-invertible symmetries as a result of flat gauging non-Abelian Lie groups.

This formulation of continuous symmetry categories invites a wide range of generalizations. For systems in higher dimensions, we expect to have an $n$-category $n\hilb^\kappa(G)$. There has been various constructions for $n\hilb$ category~\cite{Baez:1995xq,Baez:2008hz}, but there is no existing literature describing the ``$G$-graded'' version. Even more generically, invertible symmetries forms a higher-group structure~\cite{Gaiotto:2014kfa,Cordova:2018cvg,Liu:2024znj}, the higher representation theory of continuous higher-group symmetries will become handy when dealing with QFTs with Green-Schwarz mechanism.

For fermionic systems, it would be interesting to incorporate the fermionic parity symmetry and the related Gu-Wen type anomalies, in which case the symmetry category is expected to be ${\rm sHilb}^\mu(G)$, the SymTFT and corresponding boundary conditions could also be formulated~\cite{Bhardwaj:2024ydc,Wen:2024udn}.

For open quantum systems where symmetries are decohered, there has been constructions of open SymTFT to classify mixed-state phases with strong and weak symmetries and symmetry breakings~\cite{Luo:2025phx,Schafer-Nameki:2025fiy,Qi:2025tal}, the same could be expected for continuous SymTFT.

We also put forward a series of open questions for mathematicians. During our quest for $\CZ(\hilb^k(G))$, we did not recover the entire braided tensor category structure, thus
\begin{question}
    What is the complete categorical structure of $\CZ(\text{Hilb}^k(G))$ within the framework of continuous categories? Furthermore, what is the rigorous analytic generalization of modularity that accommodates the continuous spectrum of simple objects in this center?
\end{question}
With answer to this question, we should consider taking the category to a $3$D TQFT, which physicists tend to formulate as a $BF+k CS$ theory. The mathematical conundrum is the dualizability requirement in the target $3$-category when we build the fully-extended TQFT, so:
\begin{question}
    Can we construct a $3$D (fully-)extended TQFT $Z$ assigning $Z(S^1) = \CZ(\text{Hilb}^k(G))$ by choosing the target 3-category to incorporate appropriate analytic data? If so, what replaces the standard finiteness obstructions of fully dualizable objects in this continuous regime?
\end{question}
With the $3$D TQFT $Z$ well-defined, we intend to classify the topological boundary conditions (gapped boundaries) of $Z$ as well as the conformal boundary conditions (gapless boundaries), leading to the question:
\begin{question}
    What is the classification of possible topological and conformal boundary conditions for $Z$?
\end{question}
Ultimately, resolving these questions requires the mathematical formulation of TQFTs to transcend certain conditions enshrined in the standard dualizable framework, necessisating a analytic framework facilitated by functional analysis and operator algebras.

\acknowledgments
RL thanks Bin Gui for useful discussions. QJ is supported by National Research Foundation of Korea (NRF) Grant No. RS-2024-00405629 and Jang Young-Sil Fellow Program at the Korea Advanced Institute of Science and Technology. RL and YNW are supported by National Natural Science Foundation of China under Grant No. 12422503. JT is supported by National Natural Science Foundation of China under Grant No. 12405085 and by the Natural Science Foundation of Shanghai (Grant No. 24ZR1419300). JT would also like to thank Ying Zhang for her love and support. YZ is supported by WPI Initiative, MEXT, Japan at Kavli IPMU, the University of Tokyo.

\appendix

\section{Basic Operator Algebra}\label{app:OperA}

In this appendix we present the basic mathematical knowledge on \calg and von Neumann algebra needed in this paper.

Most quantum mechanics courses start with states and Hilbert spaces, then construct observables as linear operators. The intent of operator algebra is to take \emph{the algebra of observables} in a quantum system to be basic, and build Hilbert spaces through the representation theories of the operator algebra~\cite{Gelfand1987OnTE,Farah:2019,Takesaki:1979}.

\begin{defn}
    A {\rm $C^*$ algebra} $A$ is an algebra over $\bbC$ equipped with
    \begin{itemize}
        \item an antilinear map $*:A\to A$ s.t.
        \be (a^{*})^* = a \ , \ (ab)^* = b^*a^* \ , \ \forall a,b\in A \, ; \ee
        \item a norm $|| \cdot ||:A\to \bbR_{\geq 0}$ s.t. \be ||ab|\le ||a||\cdot \\ ||b|| \ , \ \forall a,b\in A\,,\ee and $A$ is complete under this norm;
        \item the $C^*$ property insisting $||a^*a|| = ||a||^2$.
    \end{itemize}
\end{defn}
The most commonly used example of \calg is the linear operators on a Hilbert space.
\begin{exm}
    For a Hilbert space $\CH$, we denote $\mathscr{B}(\CH)$ the algebra of bounded linear operators on $\CH$.
\end{exm}
Another example of $C^*$ algebra that we will frequently encounter is the commutative algebra of continuous functions.

\begin{exm}
    For a compact Lie group\footnote{This definition can be extended to the case when $G$ is only locally compact. We just require continuous functions vanish at infinity.} $G$, we denote $C_0(G)$ the algebra of continuous $\bbC$-valued functions on $G$. The algebra product is given by
    \be\ba \cdot: C_0(G)\times C_0(G) &\to C_0(G) \\ (f,h)&\mapsto fh  \ , \  
    fh(g) = f(g)h(g) \ .\ea\ee
    The norm is given by the supremum norm,
    \be ||f|| =  \sup_{g\in G} |f(g)| \ . \ee
\end{exm}
In fact, there is a theorem about the algebra of continuous functions:
\begin{theorem}[Gelfand-Naimark]\label{thm:GelfandNaimark}
    Every Abelian \calg is isomorphic to $C_0(X)$ for some locally compact Hausdorff space $X$.
\end{theorem}
It is then natural to study the representation theory of \calgs, which describes how the \calg s act on Hilbert spaces.
\begin{defn}
    For a \calg $A$, a {\rm (left) $A$-module} is a Hilbert space $\CH$ with a *-homomorphism (a linear/antilinear map between *-algebras that preserves the algebra multiplication and the $*$ operation) from $A$ to the algebra of bounded operators on $\CH$, denoted $\mathscr{B}(\CH)$
    \be  \pi: A \to \mathscr{B}(\CH) \ . \ee
\end{defn}
In many literatures, a left module defined above is also denoted a representation of \calg.

Before discussing a specific example of \calg representation, we need to define one more concept:
\begin{defn}[Measurable Hilbert Bundle]
Let $(X, \mu)$ be a measurable space. A {\rm measurable Hilbert bundle} over $X$ is a family of separable complex Hilbert spaces $\{\mathcal{H}_x\}_{x \in X}$ together with a set $\mathcal{M}$ of sections (functions $s: X \to \prod_{x \in X} \mathcal{H}_x$ such that $s(x) \in \mathcal{H}_x$) s.t. the following property holds:
\begin{enumerate}
    \item For every $m\in \CM$, the function $||m(x)||$ is a measurable function on $X$.
    \item For every $m,n \in \mathcal{M}$, the function $x \mapsto \langle m(x), n(x) \rangle_{\mathcal{H}_x}$
    is a measurable function on $X$.
    \item There is a countable subset $\CM^{\prime} \subset \CM$ such that for any $x \in X$, the closure of the span of vectors $m(x)\left(m \in \CM^{\prime}\right)$ coincides with $H_x$.
\end{enumerate}
The space of square-integrable sections, denoted by
\be \int_X^\oplus \mathcal{H}_x \, d\mu(x) \ee
consists of measurable sections $s \in \mathcal{M}$ such that $\int_X \|s(x)\|^2 \, d\mu(x) < \infty$.
\end{defn}

\begin{exm}[${\rm Hilb}(G)$]\label{eg:HilbG}
    We denote $\Rep(C_0(G))$ the set of all representations of $C_0(G)$. Note that the \calg $C_0(G)$ is commutative, and it has been established~\cite{Takesaki:1979,Farah:2019} that the Hilbert space of $C_0(G)$ representation shall always take the form of
    \be
    \CH \cong \int_G^\oplus  \CH_g \  d\mu(g) \ ,
    \ee
    where $d\mu$ is a Radon measure\footnote{A Radon measure is a Borel measure that is always finite on compact spaces.} on $G$ (note that this is generically not the Haar measure) and $\CH_g$ is a measurable Hilbert bundle on $G$. Each vector in $\CH$ takes the form
    \be \xi = \int_G  \xi_g \ d\mu(g) \  , \ \xi_g \in \CH_g \ .\ee
    The action of $C_0(G)$ on $\CH$ is given by pointwise\footnote{In some literatures, such property is referred to as ``decomposable''.} multiplication
    \be \big(\pi(f)\cdot \xi\big)_g = f(g)\xi_g \ .  \ee
    We denote $\Rep(C_0(G))$ the set of all representations of $C_0(G)$.

    From this set, we can build a tensor category, denoted ${\rm Hilb}(G)$:
    \begin{itemize}
        \item ${\rm Obj}\big({\rm Hilb}(G)\big) = \Rep\big( C_0( G) \big)$. Each object is determined by a Radon measure and a measurable Hilbert bundle on $G$, denoted $(d\mu,\CH_g)$;
        \item Simple objects are given by $(\delta_x, \CL_g)$, where $\delta_x$ is a Dirac measure on $G$ and $\CL_g$ is a rank-$1$ Hilbert bundle, which gives a Hilbert space
        \be  \CH \cong \int_G^\oplus  \delta_x(g) \  \CL_g  = \CL_x \cong \bbC\ee
        For simplicity, we will refer the simple object as $\delta_x$.
        \item Morphisms between two representations are bounded intertwiners;
        \item There is a tensor product structure given by convolution on both Hilbert bundles and the measure,
        \be \ba
        m: G\times G\to G \ &, \ (g,h)\mapsto gh \\
        d\mu^*\nu &:= m_*(d\mu\times d\nu)\\
        (d\mu, \CH_g) \otimes (d\nu,\CK_g) &= (d\mu^*\nu, (\CH*\CK)_g ) \\
        (\CH*\CK)_g &= \int_G^{\oplus} d\lambda_g(h,h')\ \CH_{h}\otimes \CK_{h'}
        \ea \ee
        where $d\lambda_g$ is the conditional measure of product measure $d\mu\times d\nu$ supported on $m^{-1}(g)\subset G\times G$. For the simple objects, this convolution reads
        \be  \delta_x \otimes \delta_y = \delta_{xy} \ .  \ee
    \end{itemize}
\end{exm}

The definition requires \calgs to be norm complete, therefore the topology of $\mathscr{B}(\CH)$ is imperative. There are 3 topologies one can take: norm topology, strong operator topology (SOT) and weak operator topology (WOT).
\begin{defn}
    Suppose $\CH$ is a separable Hilbert space and $\mathscr{B}(\CH)$ the algebra of bounded operators. Suppose there is a sequence $\{a_n\}_{n\in \bbN}\subseteq \mathscr{B}(\CH)$.

    The sequence $\{a_n\}_{n\in \bbN}$ converges to $a$ in norm topology if and only if
    \be   \lim_{n\to \infty}||(a-a_n)|| = 0 \ .  \ee

    The sequence $\{a_n\}_{n\in \bbN}$ converges to $a$ in strong operator topology (SOT) if and only if
    \be  \forall \ket{\psi}\in \CH, \ \lim_{n\to \infty}||(a-a_n)\ket{\psi}|| = 0 \ .  \ee

    The sequence $\{a_n\}_{n\in \bbN}$ converges to $a$ in weak operator topology (WOT) if and only if
    \be  \forall \ket{\psi},\ket{\phi}\in \CH, \ \lim_{n\to \infty}||\bra{\phi}(a-a_n)\ket{\psi}|| = 0 \ .  \ee
\end{defn}
Suppose $B$ is a subalgebra of $\mathscr{B}(\CH)$, then taking closures under different topologies would give a sequence of inclusions
\be  \overline{B}^{||\cdot ||} \subseteq \overline{B}^{\rm SOT} \subseteq  \overline{B}^{\rm WOT}  \ . \ee
With this we can define a von Neumann algebra.
\begin{defn}
    A von Neumann algebra is a unital $C^*$-subalgebra of $\mathscr{B}(\CH)$ that is closed under SOT. A $W^*$-algebra is a \calg which is isomorphic to a von Neumann algebra.
\end{defn}
Another popular definition of von Neumann algebra by double commutant is entailed in the following theorem:
\begin{theorem}[Double commutant theorem]
    Suppose $A\subseteq \mathscr{B}(\CH)$ is a $C^*$-subalgebra and $a\in \mathscr{B}(\CH)$. Then the following statements are equivalent:
    \begin{enumerate}
        \item $a\in (A')'$, where we define the $'$ to be the commutant, $K' := \{x\in  \mathscr{B}(\CH)|\forall k\in K, \  xa =ax\} $;
        \item there exists a sequence $\{a_n\}_{n\in \bbN}$ s.t. $a = \lim_{n\to \infty} a_n$ in SOT;
        \item there exists a sequence $\{a_n\}_{n\in \bbN}$ s.t. $a = \lim_{n\to \infty} a_n$ in WOT.
    \end{enumerate}
\end{theorem}
This theorem basically gives 3 equivalent definitions of von Neumann algebra.

In the scope of this work, a relevant example of von Neumann algebra is the space of essentially bounded measurable functions.
\begin{exm}[$L^{\infty}(X,\mu)$]
    For a measurable space $X$ with measure $\mu$, we naturally have a separable Hilbert space $L^2(X,\mu)$ of square-integrable functions on $X$. The von Neumann algebra we want to consider is $L^{\infty}(X,\mu)$, the space of essentially bounded measurable functions on $X$. The norm of such space is given by the ``almost everywhere bound'',
    \be ||f|| := \inf_{C\in \bbR_+} \{   |f(x)|\le C \text{ for almost every} \ x\in X \}  \ . \ee
    This is a prototypical example of Abelian von Neumann algebra. 
    According to Theorem \ref{thm:GelfandNaimark}, $L^{\infty}(X,\mu)$ must be isomorphic to $C_0(Y)$ for some locally compact Hausdorff space $Y$. The space $Y$ turns out to be the Stone space of the quotient of the $\sigma$-algebra of $\mu$-measurable sets modulo the ideal of $\mu$-null sets~\cite{Farah:2019}.
\end{exm}

\section{Hilbert (Bi)modules, Connes Fusion and Correspondences}\label{app:bimod}
Other than the (bi)modules of \calgs, one can also consider the Hilbert (bi)modules~\cite{Kaplansky:1953kln}.
\begin{defn}
    Suppose $A$ and $B$ are \calgs, an {\rm $A$-$B$-Hilbert bimodule} is an $A$-$B$-bimodule $M$ with a sesquilinear form
    \be  \langle\cdot | \cdot \rangle : M\times M \to B \ ,  \ee
    s.t. $\forall a\in A, b\in B$ and $x,y\in M$
    \be\ba
    &\langle x | y \rangle  = \langle y | x \rangle^* \ ;  &\langle x | yb \rangle = \langle x | y \rangle b \ ;\\
    &\langle ax | y \rangle  = \langle x | a^*y \rangle  \ ;      &\langle x | x \rangle\ge 0 \ ;
    \ea\ee
    and $M$ is complete with respect to the norm $||x||= ||\expval{x|x}||^{1/2}$. We denote the set of $A$-$B$-Hilbert bimodules $\mathbf{Bim}(A,B)$.
\end{defn}

Given the bimodule structure, a natural question is if the relative tensor product
\be \otimes_B : \mathbf{Bim}(A,B) \times \mathbf{Bim}(B,C)\to \mathbf{Bim}(A,C)  \ee
could still be constructed. This is plausible through the following steps.

For two Hilbert bimodules $M\in \mathbf{Bim}(A,B)$ and $N\in\mathbf{Bim}(B,C) $, the relative tensor product $M\otimes _B N$ is constructed by the following steps (we take $m_1,m_2,m\in M, n_1,n_2,n\in N, b\in B$)
\begin{itemize}
    \item take a tensor product of algebra $M\otimes N$;
    \item the sesquilinear form is given by
    \be  \langle m_1\otimes n_1 | m_2\otimes n_2 \rangle  = \langle n_1| \expval{m_1|m_2}_M \ket{n_2}  \ ,  \ee
    note that the bracket $\expval{m_1|m_2}_M\in B$;
    \item quotient out the equivalence relation $mb\otimes n-m\otimes bn = 0$;
    \item take the norm completion.
\end{itemize}

\begin{defn}[Adjointable Operators]
Suppose $M,N$ are (right) Hilbert $A$-modules. A linear map $T: M \to N$ is said to be {\rm adjointable} if there exists a map $T^*: N \to M$ such that the following inner product identity holds for all $x \in M$ and $y \in N$:
\[
\langle Tx, y \rangle_N = \langle x, T^* y \rangle_M
\]
We denote the set of all adjointable operators from $M$ to $N$ as $\CL_A(M, N)$. If $M = N$, we write $\CL_A(M)$.
\end{defn}

\begin{defn}[$C^*$-Correspondence]
An $A-B$ {\rm correspondence} is a right Hilbert $B$-module $M$ equipped with a $*$-homomorphism 
\[
\phi: A \to \CL_A(M) \ .
\]
We denote the set of $A-B$ correspondences ${\rm Corr}(A,B)$.
\end{defn}
\begin{exm}
    A simple example of correspondence is $\Rep(A) = {\rm Corr}(A,\bbC)$.
\end{exm}
It is easy to deduce that there is a natural tensor product structure (Rieffel tensor~\cite{Rieffel1974InducedRO})
\be
-\otimes_B - : {\rm Corr}(A,B) \times {\rm Corr}(B,C) \to {\rm Corr}(A,C) \ .
\ee
Thus an $A-B$ correspondence $(M,\phi)$ naturally induces a map
\be\ba\label{eq:repbrepa}
(M,\phi): \Rep(B) = {\rm Corr}(B,\bbC) &\to {\rm Corr}(A,\bbC) =\Rep(A)\\
(\CH_g , d\mu) &\mapsto  M\otimes_B(\CH_g , d\mu) \ .
\ea\ee

\section{Continuous Tensor Category}\label{app:continuoustc}
The intuition of continuous tensor category constructed in~\cite{Marin-Salvador:2025stc} is to use the representation category of certain \calg to formulate continuously parameterized categories with monoidal structures. The data of a continuous tensor category $\CC$ is given by
\begin{itemize}
    \item A \calg $A$, which gives the information of objects ${\rm Obj}(\CC) = {\rm Obj}\big(\Rep(A)\big)$; 
    \item An $A-A\otimes A$ correspondence $(M,\phi)$ which gives tensor product
    \be
    (M,\phi) : \Rep(A) \otimes \Rep(A) \cong \Rep(A\otimes A) \to \Rep(A) \ ,
    \ee
    by left tensor product with respect to $A\otimes A$, just as in~\eqref{eq:repbrepa};
    \item A $\bbC-A$ correspondence giving the unit;
    \item An associator $\alpha$, which gives
    \be\ba
    \alpha: (M,\phi)\otimes_{A\otimes A}\Big\{ &\Big[ (M,\phi)\otimes_{A\otimes A}( R_1 \otimes R_2) \Big] \otimes R_3 \Big\} \\ &\to (M,\phi)\otimes_{A\otimes A} \Big\{R_1\otimes\Big[(M,\phi)\otimes_{A\otimes A}(R_2\otimes R_3)\Big]\Big\}
    \ea\ee
    for every $R_{1,2,3}\in \Rep(A)$, which is formally an intertwiner;
    \item Left and right unitors given by related intertwiners.
\end{itemize}

\section{Multiplicative Bundle Gerbe}\label{app:multbdgb}
In this section we introduce the concept of multiplicative bundle gerbe~\cite{Murray:1996bdg}. This can be considered a higher version of Hermitian linde bundle on a space endowed with group multiplication and coherence conditions.

We first introduce the definition of a bundle gerbe.
\begin{defn}[Bundle gerbe]
    A bundle gerbe $\mathscr{G}$ over a smooth manifold $M$ is given by the following data:
    \begin{enumerate}
        \item a surjective submersion $\varphi:Y\to M$;
        \item a $U(1)$-principle bundle
        \be L\to Y\times_M Y \,, \ee
        over the fiber product;\footnote{%
        $Y^{[2]}:=Y \times_M Y = \{(y_1, y_2) \in Y^2 \mid \varphi(y_1)=\varphi(y_2)\}$, higher fiber products are defined similarly by adding $Y$ factors.}
        \item an isomorphism of $U(1)$-bundles on $Y^{[3]} := Y\times_M Y\times_M Y$,
        \be  \rho: L_{12} \otimes L_{23} \to  L_{13}\,, \ee
        where $L_{i j}$ is the pullback of $L$ to the $(i,j)$-th factor of $Y^{[3]}$; 
        \item the isomorphism is associative over $Y^{[4]} := Y\times_M Y\times_M Y\times_M Y$, namely,
        \be \rho_{1,3,4}\circ(\rho_{1,2,3}\otimes {\rm id} )  = \rho_{1,2,4} \circ( {\rm id}\otimes \rho_{2,3,4} ) \ .\ee
    \end{enumerate}
\end{defn}
For our purpose, we take the base manifold $M$ to be the Lie group $G$. To incorporate the group multiplication structure $m:G\times G \to G$ and associativity, we require the bundle gerbe to host the corresponding structures, rendering the bundle gerbe multiplicative.
\begin{defn}[Multiplicative bundle gerbe]
A multiplicative bundle gerbe $(\mathscr{G},\mathscr{M},\alpha)$ over $G$ is a bundle gerbe $\mathscr{G}$ on $G$ with
\begin{enumerate}
    \item an equivalence of gerbes (an invertible \emph{gerbe bimodule})
    \be\label{eq:scrm} \mathscr{M}: p_1^* \mathscr{G} \otimes  p_2^* \mathscr{G}  \to m^*\mathscr{G}\,, \ee
    over $G\times G$;
    \item an isomorphism $\alpha$ implementing associativity of group multiplication, namely
    \be
    \begin{tikzcd}
\mathscr{G}_1\otimes \mathscr{G}_2\otimes \mathscr{G}_3 \arrow[d, "{{\rm id }\otimes \mathscr{M}_{2,3}}"'] \arrow[rr, "{\mathscr{M}_{1,2}\otimes {\rm id}}"] &  & \mathscr{G}_{12}\otimes \mathscr{G}_3 \arrow[d, "{\mathscr{M}_{12,3}}"] \arrow[lld, "{\alpha_{1,2,3}}", Rightarrow] \\
\mathscr{G}_{1}\otimes \mathscr{G}_{23} \arrow[rr, "{\mathscr{M}_{1,23}}"']                                                                          &  & \mathscr{G}_{123}                                                                                                      
\end{tikzcd}\ee
where we have abbreviated $p_{k}^*\mathscr{G}$ as $\mathscr{G}_k$, $m_{i,j}^*\mathscr{G}$ as $\mathscr{G}_{ij}$ and $m^*_{1,2}\circ m_{12,3}^*\mathscr{G} \cong  m_{2,3}^*\circ m_{1,23}^*\mathscr{G}$ as $\mathscr{G}_{123}$;
    \item the isomorphism shall satisfy the pentagon identity
    \be \ba &[{\rm id} \circ ({\rm id} \circ \alpha_{1,2,3})] \circ [ \alpha_{1,23,4}\circ {\rm id} ]\circ[{\rm id}\circ (\alpha_{1,2,3}\otimes {\rm id})] \\ = &[\alpha_{12,3,4}\circ {\rm id}] \circ [\alpha_{1,2,34}\circ{\rm id} ]\,, \ea\ee
    which presents two identical ways to transform
    \be
    \mathscr{M}_{123,4} \circ (\mathscr{M}_{12,3}\otimes {\rm id}_{\mathscr{G}_4}) \circ (\mathscr{M}_{1,2}\otimes {\rm id}_{\mathscr{G}_3}\otimes {\rm id}_{\mathscr{G}_4} )\,,
    \ee
    into
    \be
    \mathscr{M}_{1,234} \circ ({\rm id}_{\mathscr{G}_1}\otimes \mathscr{M}_{2,34}) \circ ({\rm id}_{\mathscr{G}_1}\otimes {\rm id}_{\mathscr{G}_2} \otimes \mathscr{M}_{3,4} ) \ .
    \ee
\end{enumerate}
\end{defn}
The classification of multiplicative bundle gerbes is given in Proposition 5.2 of~\cite{Carey:2004xt}:
\begin{prop}[Classification of multiplicative bundle gerbe]
    Let $G$ be a compact, connected Lie group. Then there is an isomorphism between $H^4(B G , \mathbb{Z})$ and the space of isomorphism classes of multiplicative bundle gerbes $(\mathcal{G},\mathscr{M},\alpha)$ on $G$.
\end{prop}

\section{Fell Bundle on Groupoids}\label{app:fellbundle}
Our construction of Drinfeld center relies heavily on the Fell bundle over groupoids and its associated $C^*$-algebra. We will introduce this notion here.

While a group mathematically formalizes the symmetries of a single object, a groupoid generalizes this concept to represent the symmetries of a family of objects.
\begin{defn}
    A groupoid $\mathcal{G}$ is a small category in which every morphism is invertible.
\end{defn}

For a groupoid $\CG$, we denote its set of objects by $\CG^{(0)} := {\rm Obj}(\CG)$ and its set of morphisms $\CG^{(1)}:= \sqcup_{x,y\in \CG^{(0)}} {\rm Hom}_\CG (x,y) $. A groupoid is equipped with structural maps including:
\begin{itemize}
    \item The source and target maps $s, t: \mathcal{G}^{(1)} \to \mathcal{G}^{(0)}$. For every $a \in {\rm Hom}_\CG (x,y)$, we have $s(a) = x$ and $t(a) = y$.
    \item An inversion map. For every $a \in {\rm Hom}_\CG (x,y)$, the inversion map gives $a^{-1}\in {\rm Hom}_\CG (y,x)$ s.t. $a\circ a^{-1} = {\rm id}_y$ and $a^{-1}\circ a  = {\rm id}_x$.
\end{itemize} 
For $x\in \CG^{(0)}$, we will denote $\CG_x:=\{ k\in \CG | s(k) = x \}$, $\CG^y :=\{ k\in\CG | t(k) = y  \}$ and $\CG_x^y :=  \CG_x\cap \CG^y$. We will also denote the space of composable pairs of morphisms $\CG^{(2)} := \{  (a,b)\in \CG^{(1)}\times \CG^{(1)} | s(a) = t(b) \}$.

In the context of Lie groupoid, we require $\mathcal{G}$ to be a locally compact Hausdorff topological groupoid, which means both $\CG^{(0)}$ and $\CG^{(1)}$ is endowed with a topology s.t. they are both locally compact Hausdorff spaces, and the source and target maps $s,t:\CG^{(1)}\to \CG^{(0)}$ are locally proper open maps.  To perform analysis and integration, $\mathcal{G}$ is equipped with a smooth left Haar system. 
\begin{defn}
    A smooth left Haar system on a Lie groupoid $\mathcal{G}$ with a base manifold $\CG^{(0)}$ is a family of smooth Radon measures\footnote{Radon measures are Borel measures with local finiteness (evaluation on compact spaces are finite).} $\{\mu^x\}_{x \in \CG^{(0)}}$ supported on the target fibers $\mathcal{G}^x$ with the following properties:
\begin{itemize}
    \item Smoothness: For any compactly supported smooth function $f \in C_c^\infty(\mathcal{G})$, the assignment $x \mapsto \int_{\mathcal{G}^x} f(\gamma) d\mu^x(\gamma)$ defines a smooth function on the base manifold $\CG^{(0)}$.
    \item Left-invariance: The measures are strictly invariant under left translation. For any morphism $a \in {\rm Hom}_{\CG}(s(a),t(a))$, and any compactly supported smooth function $f$ on $\CG$, the integral satisfies the identity $\int_{\mathcal{G}^{s(a)}} f(\eta \gamma) d\mu^{s(a)}(\gamma) = \int_{\mathcal{G}^{t(a)}} f(\gamma) d\mu^{t(a)}(\gamma)$.
\end{itemize}
\end{defn}

\begin{exm}[$G//_{\rm Ad}G$]
    The groupoid $\CG = G//_{\rm Ad}G$ is given by the following data:
    \begin{itemize}
        \item $\CG^{(0)} = G$;
        \item ${\rm Hom}_{\CG}(g,h) = \{ x\in G| xgx^{-1} = h \}$, for any morphism $k\in {\rm Hom}_{\CG}(g,kgk^{-1})$, we denote it as $(g,k)$;
        \item Multiplication $(xgx^{-1},y)\cdot (g,x) = (g,xy) \ , \ \forall g,x,y\in G$.
    \end{itemize}
\end{exm}

A Fell bundle can be thought as a bundle on the groupoid space, with additional $C^*$-algebraic structures~\cite{Fell:1969eml,Kumjian:1998umx}. To concretely define the Fell bundle, we need to introduce some analytic notions.
\begin{defn}[\cite{Muhly:2008pdm}]
    An upper semicontinuous Banach bundle over a topological space $X$ is a topological space $\mathcal{A}$ together with a continuous, open surjection $p : \mathcal{A} \to X$ and complex Banach space structures on each fiber $\mathcal{A}_x := p^{-1}(\{x\})$ satisfying the following axioms:
    \begin{enumerate}
    \item The map $a \mapsto \|a\|$ is upper semicontinuous from $\mathcal{A}$ to $\mathbb{R}^+$ (i.e., for all $\epsilon > 0$, $\{ a \in \mathcal{A} : \|a\| \ge \epsilon \}$ is closed).
    \item If $\mathcal{A}^{(2)} := \{ (a,b) \in \mathcal{A} \times \mathcal{A} : p(a) = p(b) \}$, then $(a,b) \mapsto a + b$ is continuous from $\mathcal{A}^{(2)}$ to $\mathcal{A}$.
    \item For each $\lambda \in \mathbb{C}$, $a \mapsto \lambda a$ is continuous from $\mathcal{A}$ to $\mathcal{A}$.
    \item If $\{a_i\}$ is a net in $\mathcal{A}$ such that $p(a_i) \to x$ and such that $\|a_i\| \to 0$, then $a_i \to 0_x$ (where $0_x$ is the zero element in $\mathcal{A}_x$).
\end{enumerate}
\end{defn}
After incorporating the $C^*$-structure, we can define the Fell bundle below.

\begin{defn}[Fell bundle~\cite{Muhly:2008pdm}]
    A Fell bundle is a upper semicontinuous Banach bundle $p:\CA\to \CG$, where $\CG$ is a locally compact Hausdorff groupoid with the following structure:
    \begin{enumerate}
        \item A continuous bilinear associative multiplication $m:\CA^{(2)}\to \CA$ and a continuous involution $*:\CA\to \CA$ s.t.
        \be\ba  &\forall a,b\in \CA \ , \\ &p(m(a,b)) = p(a)p(b)  \\  &m(a,b)^* = m(b^*,a^*) \\ &p(b^*) = p(b)^{-1} \ .  \ea\ee
        \item For each $u\in \CG^{(0)}$, the Banach space $p^{-1}(u)$ is a $C^*$-algebra.
        \item For each $\xi \in \CG^{(1)}$, $p^{-1}(\xi)$ is a $p^{-1}(t(\xi) )-p^{-1}(s(\xi))$ correspondence.
    \end{enumerate}
\end{defn}
\begin{remark}
A Fell bundle can be equivalently formulated as a groupoid in the category $\mfr{Corr}$, where ${\rm Obj}(\mfr{Corr}) = \{\text{\calgs}\}$ and ${\rm Hom}_{\mfr{Corr}}(A,B) = {\rm Corr}(A,B)$, the set of $A-B$ correspondences introduced in Appendix~\ref{app:bimod}.
\end{remark}

The transition from the geometric bundle to a global analytic object is achieved by constructing a convolution algebra. We define the space of continuous, compactly supported sections of the bundle denoted $\Gamma_c(\mathcal{G}, \mathcal{A})$.

This space is endowed with a *-algebra structure using the left Haar system. For any sections $f, g \in \Gamma_c(\mathcal{G}, \mathcal{A})$, the convolution product is defined via integration over the groupoid fibers:
\be (f * g)(x) = \int_{\mathcal{G}^{t(x)}} f(y)g(y^{-1}x) \ d\lambda^{t(x)}(y) \ .\ee
The involution on sections is defined by intertwining the bundle's intrinsic involution with the groupoid's inversion map
\be f^*(x) = f(x^{-1})^*\ . \ee

The universal $C^*$-algebra of the Fell bundle is denoted $C^*(\mathcal{A}, \mathcal{G})$, and is defined as the completion of the $*$-algebra $\Gamma_c(\mathcal{G}, \mathcal{A})$ with respect to the universal $C^*$-norm, which can be obtained by taking the supremum over all continuous non-degenerate $*$-representations on Hilbert spaces.

Suppose we have a Fell line bundle $\CA\to \CG$ where the fiber over both objects and morphisms of the groupoid are isomorphic to $\bbC$, the Fell bundle can be modeled by a $U(1)$-central extension of $\CG$. 
\begin{defn}[$U(1)$-Central extension of groupoid~\cite{suzuki:2018cle}]\label{def:centextu1}
    A $U(1)$-central extension of a groupoid $\CG$ is defined as a groupoid $\widetilde{\CG}$ s.t.
    \begin{itemize}
        \item The objects are given by $\widetilde{\CG}^{(0)} =\CG^{(0)}$.
        \item The set of morphisms is given by a $U(1)$-principal bundle $\pi: \widetilde{\CG}^{(1)} \to  \CG^{(1)}$, and therefore inherits a right $U(1)$ action.
        \item Suppose $z_1,z_2\in U(1)$ and $(a_1,a_2)\in \widetilde{\CG}^{(2)}$ are composable morphisms, then $(a_1z_1) \circ (a_2z_2) = (a_1\circ a_2)(z_1z_2)$.
    \end{itemize}
\end{defn}
In certain mathematical literatures such $U(1)$-central extension is also abbreviated as a sequence~\cite{Suzuki:2002hrm}
\be
  U(1) \stackrel{}{\hookrightarrow} \widetilde{\CG} \stackrel{}{\longrightarrow} \CG  \ .
\ee
For a $U(1)$-central extension $\widetilde{\CG}$ over $\CG$, a Fell line bundle $p:\CA\to \CG$ can be constructed as
\be  \CA := \widetilde{\CG} \times \bbC  / \{ (a ,\lambda)\sim (az,z^{-1}\lambda)  \} \ . \ee
The multiplication and involution are given by
\be \ba &\forall (\gamma_1,\gamma_2)\in \CG^{(2)} \ , \ a_i\in p^{-1}(\gamma_i) \ (i\in \{ 1,2\} ) \\
&m((a_1,z_1) , (a_2,z_2)) = (a_1\circ a_2 , z_1z_2) \\
&(a_1,z_1)^* = (a_1^{-1},\overline{z_1}) \ .
\ea\ee

Such central extensions are classified up to isomorphism by the second groupoid cohomology group $H^2(\mathcal{G}, U(1))$\footnote{The groupoid cohomology of groupoid $\mfr{X}$ can be constructed as the $\check{\rm C}$ech cohomology of the simplicial space $X_\bullet$, where $X$ is an atlas of $\mfr{X}$ with a representable submersion $X\to \mfr{X}$. The second cohomology class $H^2(\CG,U(1))$ characterizes $U(1)$-gerbe over $\CG$, which is different from $U(1)$-central extension over $\CG$. However, when $H^1(G,U(1)) = H^2(G,U(1))= 0$, $H^2(\CG,U(1))$ is isomorphic to the isomorphism class of $U(1)$-central extensions.}~\cite{Behrend:2008}. Choose a representative $\nu$ of the class $[\nu]\in H^2(\CG,U(1))$, the corresponding \calg of the Fell bundle admits a twisted multiplication given by the following procedure:
\begin{enumerate}
    \item For each $g\in G$, we can take a section $\theta_g : \CG_g\to \widetilde{\CG}_g$
    \item For each $g\in G$, the cocycle $\nu$ gives a map 
    \be\ba
    &\nu_g: \CG_g\times \CG^g \to U(1) \\ 
    &\nu_g(x,y) = \theta_{s(y)}(xy)^{-1} \theta_g(x) \theta_{s(y)}(y) \ ,
    \ea\ee
    where the second line holds because the RHS lives in $U(1) \cong U(1) \times \{s(y) \}\subset \widetilde{\CG}_{s(y)}$.
    \item The maps $\{\nu_g\}$ satisfies the cocycle identity
    \be
    \nu_g\left(x, y\right) \nu_h\left(x y, z\right)=\nu_g\left(x, y z\right) \nu_h\left(y, z\right) \ , \ \forall x \in G_g, y \in G_h^g, z \in G^h
    \ee
    \item The \calg of Fell bundle sections $C^*(\CG,\widetilde{\CG})$ has the multiplication given by the $\nu$-twisted convolution
    \be
    (f*g)(x) = \int_{\CG^{t(x)}} \nu_{s(x)}(y,y^{-1}x) f(y) g(y^{-1}x) \  d\lambda^{t(x)}(y)
    \ee
    as well as the involution
    \be
    f^*(x) = \overline{\nu_{s(x)}(x,x^{-1}) f(x^{-1})} \ .
    \ee
\end{enumerate}

From now on, we will denote the Fell line bundle determined by $[\nu]\in H^2(\CG, U(1))$ as $\Sigma_\nu$. We intend to discuss the representation category of $C^*(\CG, \Sigma_\nu)$. Theorem 2.1 of~\cite{Ionescu:2013} has classified the irreducible representations of such $C^*$-algebra, which in our case translates to the following:
\begin{corollary}
    Suppose $G$ is a locally compact Hausdorff Lie group, $\Sigma_\nu$ is a Fell line bundle over $\CG = G//_{\rm Ad} G$. For any $g\in G$, if $L$ is an irreducible representation of $C^*(C_G(g) ,\Sigma_\nu)$ ($C_G(g) = \CG_g^g$ is the centralizer of $g$ in $G$), then the induced representation ${\rm Ind}_{C_G(g)}^{\CG}L$ is an irreducible representation of $C^*(\CG ,\Sigma_\nu)$.
\end{corollary}
\begin{remark}
    Following the construction in~\cite{Ionescu:2013}, one can show that if $g' = t\big( (g,h) \big) = hgh^{-1}$, then the two induced representations ${\rm Ind}_{C_G(g)}^{\CG}L$ and ${\rm Ind}_{C_G(g')}^{\CG}L$ are isomorphic. This fact could also be observed from other literatures~\cite{Buss:2026}.
\end{remark}
\begin{remark}
    The generalized Effros-Hahn conjecture (partially) proved in~\cite{Ionescu:2009} states that if $\Sigma_\nu$ and $\CG$ are second countable amenable locally compact Hausdorff groupoids, then any irreducible representations of $C^*(\CG ,\Sigma_\nu)$ is always an induced representation of certain $L\in {\rm IrRep}\left(C^*\left( C_G\left(g\right),\Sigma_\nu \right)\right)$. The amenability condition for $G//_{\rm Ad} G$ is equivalent to amenability for $G$, which includes all compact groups and all Abelian groups. For non-compact cases there are non-amenable groups like $SL(2,\bbR)$.
\end{remark}

On $C_G(g) = \CG_g^g$, the Fell bundle gives a restricted $\nu_g: C_G(g)\times C_G(g)\to U(1)$ satisfying cocycle identity, which exactly labels the projective representation of $C_G(g)$.

With the above analysis, we can state the following proposition:
\begin{prop}
    Suppose $G$ is a second countable, amenable, locally compact Hausdorff Lie group, let $\CG = G//_{\rm Ad}G$ be the conjugation groupoid and let $\Sigma_\nu$ be a $U(1)$ central extension twisted by $\nu\in H^2(\CG,U(1))$. Then every irreducible representation of $C^*(\CG,\Sigma_\nu)$ is labeled by a conjugacy class $[g]$ and a projective irreducible representation of the centralizer $\rho_{\nu_g} \in {\rm IrRep}(C_G(g))$.
\end{prop}

\section{Injective Transgression}\label{app:injectivetransgression}
In this appendix we focus on one question: under what circumstances is the transgression
\be \tau: H^4(BG,\bbZ) \to H^3_G(G,\bbZ) = H^3(EG\times_G G,\bbZ) = H^3(LBG,\bbZ)\ee
an injective homomorphism?

Suppose the Lie group $G$ is compact connected, then its cohomology $H^4(BG,\bbZ)$ is torsion-free (see Theorem 6 of~\cite{Henriques:2016ipo}). This means using $\bbQ$-valued cohomology does not loose any infomation. By taking tensor with $\bbQ$, we get the rational transgression
\be  \tau \stackrel{\otimes \bbQ}{\longrightarrow} \tau_\bbQ : H^4(BG,\bbQ) \to H^3(LBG,\bbQ) \ . \ee

For $\bbQ$-valued cohomology, there is a convenient tool known as the Sullivan model~\cite{ViguPoirrier1976TheHT}. The idea of Sullivan model is to take a commutative differential graded algebra (cdga) whose cohomology reproduces $H^\bullet(X,\bbQ)$ for some space $X$. In the following, we will assume $V$ to be a graded vector space, and denote $V^+:= \bigoplus_{k\ge 1} V^k $. We also let $\Lambda^k V$ denote the space spanned by $\{ v_1\wedge \cdots \wedge v_k|v_i\in V \}$, and $\Lambda V= \bigoplus_k \Lambda^k V$.
\begin{defn}[Relative Sullivan model]
    A relative Sullivan model is a morphism between commutative differential graded algebras (cdgas)
    \be  (B,d_B) \hookrightarrow (B\otimes \Lambda V,d) \ , \ b\mapsto b\otimes 1 \ . \ee
\end{defn}
\begin{defn}
    A Sullivan model is a relative Sullivan model with
    \be   (B,d_B) = (\bbK , 0)  \hookrightarrow (\Lambda V,d) \ . \ee
\end{defn}
For connected compact Lie group $G$, the rational cohomology of $BG$ admits a polynomial presentation
\be
H^\bullet(BG,\bbQ) \cong \bbQ[x_1,\cdots ,x_r] \ , \ {\rm deg}(x_i)\in 2\bbZ
\ee
and the Sullivan model is 
\be (\Lambda V,d) = (\bbQ[x_1,\cdots x_r],0)  \ . \ee
The Sullivan model of the free loop space $LBG$ is given by (See Section 15 of~\cite{Felix:2001rht})
\be (\Lambda V \otimes \Lambda(sV),D) = (\bbQ[x_1,\cdots ,x_r]\otimes \Lambda(sx_1,\cdots sx_r), 0) \ , \ee
where there is a map $S$ relating $x_i$ and $sx_i$
\be\ba &S:x_i\mapsto sx_i \ , \  {\rm deg}(sx_i) = {\rm deg}(s_i)-1\\ &S:sx_i \mapsto 0 \ , \ S(x_ix_j) = S(x_i)x_j + x_i S(x_j)  = (sx_i) x_j + x_i (sx_j) \ . \ea\ee
The free loop space cohomology is therefore
\be H^3(LBG , \bbQ) \cong  \bbQ[x_1,\cdots ,x_r]\otimes \Lambda(sx_1,\cdots sx_r) \ . \ee
\begin{prop}
    Suppose $G$ is a compact connected Lie group, the map $S:H^4(BG,\bbQ) \to H^3(LBG,\bbQ)$ is injective.
\end{prop}
\begin{proof}
    Suppose in $H^4(BG,\bbQ)\cong \bbQ[x_1,\cdots,x_r]$, the degree $2$ elements in $\{x_i\}_{i\in\bbZ\cap[1,r]}$ are denoted $\{y_i\}_{i\in \bbZ\cap [1,m]}$ and degree $4$ elements are denoted $\{z_i\}_{i\in \bbZ\cap [1,n]}$.
    A generic $\alpha\in  H^4(BG,\bbQ)$ takes the form
    \be \alpha = \sum_{a=1}^n A_a z_a + \sum_{1\le j\le k\le m}B_{jk}y_jy_k. \ee
    Applying $S$ map, we get
    \be  S(\alpha)=\sum_{a=1}^n A_a\,sz_a+\sum_{1\le j<k\le m}B_{jk}(y_j\,sy_k+y_k\,sy_j)+\sum_{j=1}^m2B_{jj}y_j\,sy_j. \ee
    If $S(\alpha) = 0$, then all $A_a = 0$ and for each $l$, we shall have
\be 2B_{\ell\ell}y_\ell+\sum_{j<\ell}B_{j\ell}y_j+\sum_{k>\ell}B_{\ell k}y_k=0.
\ee
Knowing all $y_i$s are linearly independent, we proved $S(\alpha) = 0 \Rightarrow \alpha = 0$. The map $S:H^4(BG,\bbQ) \to H^3(LBG,\bbQ)$ is thus injective.
\end{proof}

We then demonstrate that the $S$ map is actually the transgression. Consider the map
\be  {\rm ev}: S^1\times LBG \ , \ (t,\gamma) \mapsto \gamma(t) \ ,  \ee
where $t\in S^1$ and $\gamma$ is a loop on $BG$.The pullback of a cohomology class $\alpha\in H^4(BG,\bbQ)$ gives
\be {\rm ev}^* \alpha \in H^4(S^1\times LBG,\bbQ) \ . \ee
Integrating out along the $S^1$ circle gives our desired rational transgression~\cite{waldorf2016transgression}
\be \tau_\bbQ(\alpha) = \int_{S^1} {\rm ev}^*\alpha \in H^3(LBG,\bbQ) \ . \ee
How does this process look on the Sullivan model? A Sullivan model for $S^1$ is just $(\Lambda (u),0)$, so the Sullivan model for $S^1\times LBG$ is \be (\Lambda(u),0)\otimes (\bbQ[x_1,\cdots x_r]\otimes  \Lambda [sx_1,\cdots,sx_r],0)\ , \ee
giving the pullback of evaluation
\be {\rm ev}^* : v \mapsto  v+ u \, S(v) \ .\ee
Evaluation around $S^1$ is merely to take the coefficient of $u$, thus
\be \int_{S^1} {\rm ev}^* (\alpha) = S(\alpha) = \tau(\alpha) \ .  \ee
Thus we have:
\begin{prop}
    Suppose $G$ is a compact connected Lie group, the rational transgression $\tau_\bbQ:H^4(BG,\bbQ) \to H^3(LBG,\bbQ)$ is injective.
\end{prop}
Since there is no torsion part in $H^4(BG,\bbQ)$, we are safe to claim:
\begin{prop}
    Suppose $G$ is a compact connected Lie group, the integral transgression $\tau:H^4(BG,\bbZ) \to H^3(LBG,\bbZ) = H^3_G(G,\bbZ)$ is injective.
\end{prop}

\bibliographystyle{JHEP}
\bibliography{biblio.bib}

\end{document}